\documentclass[preprint,12pt]{elsarticle}

\usepackage[utf8]{inputenc}
\usepackage{amsfonts}
\usepackage{amssymb}
\usepackage{amsthm}
\usepackage{amsmath}
\usepackage{mathrsfs}
\usepackage[letterpaper,left=1in,right=1in,top=1in,bottom=1.2in,footskip=0.6in]{geometry}
\usepackage[all]{xy}
\usepackage{marginnote}

\usepackage{hyperref}
\hypersetup{colorlinks=true,citecolor=blue,linkcolor=blue,urlcolor=blue}

\usepackage{eurosym}
\usepackage{graphicx}
\usepackage{colortbl}
\usepackage{epstopdf}
\usepackage{array}

\newcommand*{\Z}{\mathbb{Z}}

\newtheorem{theorem}{Theorem}
\newtheorem{lemma}[theorem]{Lemma}
\newtheorem{definition}[theorem]{Definition}
\newtheorem{corollary}[theorem]{Corollary}
\theoremstyle{definition}
\newtheorem{remark}{Remark}
\newtheorem{example}{Example}

\newcounter{claimcount}
\newenvironment{claim}{\refstepcounter{claimcount}\textbf{Claim \arabic{claimcount}:}}{}

\usepackage{etoolbox}
\AtBeginEnvironment{proof}{\setcounter{claimcount}{0}}

\usepackage{tikz}
\usetikzlibrary{calc}
\usepackage{nicematrix}
\usepackage{verbatim}
\usepackage{fancyhdr}
\DeclareFontFamily{U}{mathx}{}
\DeclareFontShape{U}{mathx}{m}{n}{<-> mathx10}{}
\DeclareSymbolFont{mathx}{U}{mathx}{m}{n}
\DeclareMathAccent{\widehat}{0}{mathx}{"70}
\DeclareMathAccent{\widecheck}{0}{mathx}{"71}
\usepackage{arydshln}
\usepackage{titlesec}

\journal{Information and Computation}

\begin{document}

\begin{frontmatter}



\title{Eulerian orientations and Hadamard codes:\\ A novel connection via counting\tnoteref{tn1}}
\tnotetext[tn1]{An extended abstract of this paper appeared in the Proceedings of the 16th Innovations in Theoretical Computer Science (ITCS 2025)~\cite{shao_et_al_eo}, \url{https://doi.org/10.4230/LIPIcs.ITCS.2025.86}.
}


\author{Shuai Shao\fnref{fn1}\corref{cor1}} 
\ead{shao10@ustc.edu.cn}
\cortext[cor1]{Corresponding author}
\affiliation{organization={School of Computer Science and Technology \& Hefei National Laboratory, \\ University of Science and Technology of China},
            city={Hefei},
            country={China}}

\author{Zhuxiao Tang\fnref{fn2}}
\ead{zztang@wisc.edu}
\fntext[fn1]{Supported by the Quantum Science and Technology $-$ National Science and Technology Major Project (QNMP), 2021ZD0302901,  and the National Natural Science Foundation of China, No. 62572452.}
\fntext[fn2]{This work was done while the second author was an undergraduate student in the School of the Gifted Young, University of Science and Technology of China, under the support by QNMP 2021ZD0302901.}
\affiliation{organization={Department of Computer Sciences, University of Wisconsin-Madison},
            city={Madison},
            country={U.S.A.}}

\begin{abstract}
We discover a novel connection between two classical mathematical concepts---Eulerian orientations and Hadamard codes, by studying the 
 counting problem of Eulerian orientations (\#EO)  with local constraint functions imposed on vertices. 
    We present two special classes of constraint functions and a chain reaction algorithm, and show that the \#EO problem defined by  each class alone is polynomial-time solvable by the algorithm. 
    These tractable classes of functions are defined inductively, and quite remarkably, the base level of these classes is  characterized precisely by the well-known Hadamard code.
    Thus, we establish a novel connection between counting Eulerian orientations and coding theory. 
    We also prove a \#P-hardness result for the \#EO problem when constraint functions from the two tractable classes appear together. 
\end{abstract}



\begin{keyword}
Eulerian orientations \sep Hadamard codes \sep Counting problems \sep Tractable classes



\end{keyword}

\end{frontmatter}



\section{Introduction}
   The notion of Eulerian orientations arises from the historically notable \emph{Seven Bridges of Königsberg} problem, whose solution by Euler in 1736~\cite{euler1741solutio} is considered one of the first result of graph theory. 
  Given an undirected graph $G$,  an Eulerian orientation  of $G$ is an assignment of a direction to each edge of $G$ such that at each vertex $v$, the number of incoming edges is equal to the number of outgoing edges.
  A connected graph has an Eulerian orientation, called an Eulerian graph if and only if every vertex has even degree. 
  This is the well-known Euler's theorem, whose first complete proof  was given back in the 1800s~\cite{biggs1986graph}. 
  In terms of computational complexity,  Euler's theorem implies that the decision problem of determining whether a graph has an Eulerian orientation is polynomial-time solvable (tractable). 
  While for the counting problem,  Mihail and Winkler showed that  counting the number of Eulerian orientations of an undirected graph is \#P-complete in 1996~\cite{MW}---over a century after Euler's theorem.
  However, the counting problem becomes tractable when certain particular restrictions are imposed on edges.
  An intriguing example of such tractable problems comes from computing the partition function  of the six-vertex model~\cite{Pauling, Slater, rys1963uber}, one of the most intensively studied models in statistical physics.  
  In this example, the graphs are 4 regular and on each vertex exactly one edge incident to it is restricted to take the direction coming into the vertex. 
  Then  counting the number of Eulerian orientations obeying  restrictions on these edges is shown to be tractable~\cite{cfx}. 

  In this paper, we further study the counting restricted Eulerian orientations (\#EO) problem,  defined by constraint functions
placed at each vertex that represent 
restrictions on edges incident to the vertex.  
We consider which classes of constraint functions make the \#EO problem tractable. 
Before we formally define the problem and present our results, we first  take a detour to another classical notion from coding theory, the Hadamard code. 
  The Hadamard code is an error-correcting code used for error detection and correction when transmitting messages over very noisy or unreliable channels. 
  Because of its rich mathematical properties, 
  the Hadamard code is not only used in coding and communication, but also  studied extensively in various areas of mathematics and theoretical computer science. 
  However, it has not been known how the Hadamard code can be related to  Eulerian orientations, and particularly, how it can be used to carve out  tractable classes for the \#EO problem. 
  Quite surprisingly, in this paper, we establish a novel connection between these two well-studied mathematical notions. 
  We present new tractable classes  for the \#EO problem and a corresponding algorithm based on a chain reaction. 
  The core components of these classes, referred to as kernels, are characterized perfectly by the Hadamard code (or more precisely, the balanced Hadamard code).

  Now we formally define the \#EO problem.
   A 0-1 valued constraint function (or signature) $f$ of arity $n$ is a map $\mathbb{Z}^n_2 \to \{0, 1\}$. 
   The support of  $f$  is $\mathscr{S}(f) = \{\alpha \in \mathbb{Z}_2^n \mid f(\alpha)=1\}$.
   A signature $f$ of arity $2n$ is an \emph{Eulerian orientation} (EO) signature
 if for every $\alpha\in \mathscr{S}(f)$, the Hamming weight wt$(\alpha)=n$.
The problem \#EO($\mathcal{F}$) specified by a set $\mathcal{F}$ of EO signatures is defined as follows.
The input is a graph $G$ where 
each vertex $v$ of $G$ is associated with  some function $f_v$ from $\mathcal{F}$.
The 
incident edges  to  $v$  are totally ordered and correspond to input variables
to $f_v$.
Each edge has two ends, and an  orientation of the edge is  denoted by assigning 
0 to the head and 1 to the tail.  
Thus, locally at every vertex $v$,  a
$0$ represents an incoming edge and a $1$ represents an outgoing edge. 
An Eulerian orientation corresponds to an assignment to each edge ($01$ or $10$) where the numbers of 0s and 1s at each $v$ are equal. 
The local constraint function $f_v$ evaluates to 1 if the local assignment on edges incident to $v$ satisfies the restriction imposed by $f_v$; otherwise, $f_v = 0$.
Since each $f_v$ is an EO signature, it evaluates to 1 only when the numbers of input 0s and 1s are equal---hence only Eulerian orientations contribute nonzero values.
An Eulerian orientation contributes value $1$  if $f_v=1$ for every vertex $v$.
The \#EO problem  outputs the number of Eulerian orientations  contributing value $1$. 

 \begin{example}
Let $\mathcal{F}=\{g_2, g_4, \ldots g_{2n}, \ldots\},$ where $\mathscr{S}(g_{2n})=\{\alpha \in \mathbb{Z}_2^{2n} \mid
 {\rm wt}(\alpha)= n \}$.
Then {\rm  \#EO}$(\mathcal{\mathcal{F}})$  counts the total number of all Eulerian orientations of the input graph, which is \#P-hard.
\end{example}

\begin{example}\label{example-2}
    Let $f_2$ be a signature of arity $4$ with support $\mathscr S(f_2)=\{1100, 1010, 1001\}$.
    Then, {\rm  \#EO}$({\{f_2\}})$ computes the partition function of a six-vertex model which is tractable. 
The reason for naming this function $f_2$ will be explained in Section~\ref{sec: characterization}.
\end{example}

The \#EO problem has an intrinsic significance in 
counting problems. 
In~\cite{cai2020beyond}, the framework of  \#EO problems is formally introduced and is shown to be expressive enough to encompass all Boolean counting constraint satisfaction problems (\#CSP)~\cite{creignou1996complexity, dyer2009complexity, bulatov2009complexity, Cai-Lu-Xia-csp} with arbitrary constraint functions that are not necessarily supported on half weighted inputs, although the \#EO problem itself requires all constraint functions to be supported on  half weighted inputs.
Also, the \#EO problem encompasses many natural problems from statistical physics and combinatorics, such as the partition function of the six-vertex model  and the  evaluation of the Tutte polynomial
at the point $(3, 3)$~\cite{tutte}.
Cai, Fu and Xia proved a complexity classification for computing the partition function
of the six-vertex model on general 4-regular graphs \cite{cfx}.

Beyond these interesting concrete problems expressible as \#EO problems, 
 the study of the \#EO framework plays a crucial role in a broader picture, the complexity classification program for a much more expressive counting framework,  the Holant problem.
 Significant progress has been made in  the complexity classification of Holant problems~\cite{cai2011dichotomy, cai2016complete, Backens-holant-plus, Cai-Lu-Xia-holant-c, Backens-Holant-c, beida, cai2020holant}, culminating in a full complexity dichotomy for all real-valued Holant problems~\cite{real-holant}.
This dichotomy relies fundamentally on a  complexity dichotomy for the \#EO problem~\cite{cai2020beyond} with complex-valued signatures assuming a physical symmetry called  \emph{arrow reversal symmetry}  ({ARS}) since 
under a suitable holographic transformation, the {ARS} property corresponds
to precisely real-valued signatures. 
In this paper, we consider the 0-1 valued \#EO problem without assuming {ARS}, which may
serve as a building block toward the ultimate classification for all  complex-valued Holant problems. 
For this setting, we present new tractable classes that lie outside all previously known tractable classes for Holant problems.

 We use $\text{\sc arity}(f)$ to denote the arity of $f$, and $I(f)=\{1,2,...,\text{\sc arity}(f)\}$ to denote the indices of variables of $f$.
  We use $\delta_1$ and $\delta_0$ to denote the unary signatures where $\mathscr S (\delta_1)=\{1\}$ and $\mathscr S (\delta_0)=\{0\}$, respectively. 
  Pinning the $i$-th variable of $f$ to $0$ gives the signature
     $f(x_1,...,x_{i-1},0,x_{i+1},...,x_n)$ of arity $n-1$, denoted by $f_{i}^{0}$.
     Similarly, we define $f_{i}^{1}$.
      Let $f$ and $g$ be two signatures of support size $n$ and $m$ respectively. Suppose $\mathscr S(f)=\{\alpha_1,\alpha_2,...,\alpha_n\}$ and $\mathscr S (g)=\{\beta_1,\beta_2,...,\beta_m\}$. Then the tensor product $f\otimes g$ of $f$ and $g$ is the signature where $\mathscr S (f\otimes g)=\{\alpha_i\beta_j\mid 1\leq i\leq n, 1\leq j\leq m \}$.
        A nonzero signature $g$ is a factor of $f$, denoted by $g \mid f$, if there exists a signature $h$ such that  $f=g \otimes h$ (possibly with a permutation of variables)
or there is a constant $\lambda$ such that $f= \lambda \cdot g$.
Otherwise, $g$ is not factor of $f$, denoted by $g\nmid f$.  
In the following, we first state the definition of affine signatures, which are known to be tractable for the \#EO problem. Then we give the new tractable classes.

 \begin{definition}[Affine]\label{affine}
        A signature $f(x_1,x_2,...,x_n)$ of arity $n$ is affine if $\mathscr{S}(f)=\{x\in\mathbb{Z}_2^n | AX=0\}$, where $X=(x_1,x_2,...,x_n,1)$ and $A$ is a matrix over $\mathbb{Z}_2$. Equivalently, it satisfies the property that if $\alpha,\beta,\gamma\in \mathscr{S}(f)$, then $ \alpha\oplus\beta\oplus\gamma\in \mathscr{S}(f)$. 
    \end{definition}
    
 \begin{definition}[$\delta_1$-affine and $\delta_0$-affine]\label{def:delta1-affine}
            An EO signature $f$  is $\delta_1$-affine, denoted by $f\in \mathscr{D}_1$ if  $f = \delta_1 \otimes g$ for some $g$ where   for every $i\in I(g)$, $g_i^0\in \mathscr{A}\cup\mathscr{D}_1$.
       Symmetrically,  
       an EO signature $f$ is $\delta_0$-affine, denoted by $f\in \mathscr{D}_0$ if  $f = \delta_0 \otimes g$ for some $g$ where   for every $i\in I(g)$, $g_i^1\in \mathscr{A}\cup\mathscr{D}_0$.
    \end{definition}

 For example, the 4-ary signature $f$ in Example~\ref{example-2} is a $\delta_1$-affine signature.
Note that $\mathscr{D}_1$ (or $\mathscr{D}_0$) does not encompass $\mathscr{A}$, since an affine signature may not have a $\delta_1$ (or $\delta_0$) factor.
Thus, our tractable classes are $\mathscr{A} \cup \mathscr{D}_1$ and $\mathscr{A} \cup \mathscr{D}_0$, rather than $\mathscr{D}_1$ and $\mathscr{D}_0$.

       \begin{theorem}\label{thm:tractablity}
        {\rm \#EO($\mathscr{A} \cup \mathscr{D}_1$)} and {\rm \#EO($\mathscr{A} \cup \mathscr{D}_0$)} are tractable.
    \end{theorem}

    The tractability result is established via a chain reaction algorithm, in which the presence of a $\delta_1$ or $\delta_0$ signature plays a role similar to a neutron in a nuclear chain reaction: it initializes the reaction and generates a new neutron (a $\delta_1$ or $\delta_0$ signature, respectively) through propagation, allowing the chain reaction to continue.
    Eventually, the chain reaction terminates when a stable state---a tractable instance of $\rm \#EO(\mathscr{A})$---is reached.
    
 Although the reaction algorithm works for $\delta_1$-affine signatures or $\delta_0$-affine signatures individually, it does not work when both of them appear. 
 This is because when both $\delta_1$-affine and $\delta_0$-affine signatures are present, connecting $\delta_1$ and $\delta_0$ using $\neq_2$ causes them to annihilate each other, and no new $\delta_1$ or $\delta_0$  are generated.
 Thus, the chain reaction is unsustainable.
 This is somewhat analogous to the phenomenon of electron–positron annihilation in nuclear physics.
 In fact, we can show that the \#EO problem is \#P-hard when both types of signatures are present.
     
\begin{theorem}\label{thm:hard}
      {\rm \#EO($\mathscr{D}_0 \cup \mathscr{D}_1$)} is \#P-hard.
\end{theorem}

     Since the classes $\mathscr D_1$ and $\mathscr D_0$ are defined inductively, 
     it is not an easy task to obtain explicit expressions for them in a form similar to other known tractable classes for counting problems such as affine signatures. 
     Note that the tractability of $\mathscr D_1$ or $\mathscr D_0$ is obtained by eventually reducing the problem into a problem where all signatures are affine signatures. 
 Thus, we focus on $\delta_1$-affine (or $\delta_0$-affine) signatures, from which by pinning an arbitrary variable to $0$ (or $1$ respectively), we always get an affine signature. 
In other words, these $\delta_1$-affine (or $\delta_0$-affine) signatures lie at the first level of the inductive hierarchy. 
We define them as kernels.

 \begin{definition}[$\delta_1$-affine and $\delta_0$-affine kernel]\label{delta1-affine-kernel}
        An {\rm EO} signature $f$ is a $\delta_1$-affine kernel, denoted by $f\in K(\mathscr{D}_1)$ if  it satisfies the following two properties: 
    \begin{enumerate}
        \item $f=\delta_1^{\otimes m}\otimes h$, for some $m \in \mathbb{Z_+}$ where $h$ is a signature of positive arity, $\delta_1 \nmid h$ and $h\notin \mathscr A$;
        \item 
        for every $i\in I(h)$, $h_i^0\in \mathscr{A}$.
    \end{enumerate}
    
       Symmetrically, we define $\delta_0$-affine kernels denoted by $K(\mathscr{D}_0)$ by switching $0$ and $1$ everywhere. 
       Notice that $h\notin \mathscr A$ implies $f\notin \mathscr A$.
        Thus, $\delta_1$-affine and $\delta_0$-affine kernels are not affine. 
    \end{definition}

     A main technical contribution of the paper is a complete characterization of $\delta_1$-affine and $\delta_0$-affine kernels using 
the Hadamard code. 
 The Hadamard code refers to the code based on Sylvester's construction of Hadamard matrices. 
     A Hadamard matrix $H_n$  is a square matrix of size $n$ such that all its entries are in $\{1, -1\}$ and $H_nH^T_n = nI_n$.
Sylvester's construction gives Hadamard matrices $H_{2^k}$ of size $2^k$ for all $k\geq 0$, where $H_{2^{k+1}}=\left[
\begin{smallmatrix}
    H_{2^k} & H_{2^k}\\
    H_{2^k} & -H_{2^k}\\
\end{smallmatrix}
\right]$ and $H_{2^0}=1$.
The Hadamard code $\mathscr{H}^1_{2^k}$, derived from the Hadamard matrix $H_{2^k}$, is a binary code over ${0, 1}$ whose $2^k$ codewords are the rows of $H_{2^k}$,
with entries $1$ mapped to $1$ and entries $-1$ mapped to $0$.
Symmetrically, the Hadamard code $\mathscr{H}^0_{2^k}$ is defined by mapping $1 \mapsto 0$ and $-1 \mapsto 1$.
Note that $\mathscr H^1_{2^k}$ contains an all-$1$ codeword and  $\mathscr H^0_{2^k}$ contains an all-$0$ codeword. 
We call the former the 1-Hadamard code and the latter the 0-Hadamard code. 
In standard coding theory notation for block codes,  $\mathscr H^1_{2^k}$ and $\mathscr H^0_{2^k}$ are $[2^k, k, 2^{k-1}]$-code, which is a binary linear code having block length $2^k$, message length $k$, and minimum Hamming distance $2^k/2$.
By removing the all-$1$ codeword from $\mathscr H^1_{2^k}$, we get a balanced code, i.e., each codeword has Hamming weight $2^{k-1}$, half of its block length. 
We call it the balanced 1-Hadamard code, denoted by ${\mathscr H^{\rm 1b}_{2^k}}$.
Symmetrically, we can get the balanced 0-Hadamard code, denoted by ${\mathscr H^{\rm 0b}_{2^k}}$.
A binary code $\mathscr C_n$ of block length $n$ can be viewed equivalently as a 0-1 valued constraint function $f$ of arity $n$ by taking the support $\mathscr{S}(f) = \{\alpha \in \mathbb{Z}_2^n \mid f(\alpha) \neq 0\}$  to be $\mathscr C_n$. 
An \emph{$m$-multiple} of the code $\mathscr C_n$ is a code of block length $nm$, in which each codeword is obtained taking the concatenation of $m$ copies of $\alpha\in \mathscr C_n$.

\begin{theorem}\label{thm:characterization}
   An {\rm EO} signature $f$ is a $\delta_1$-affine  kernel if and only if 
  $\mathscr S(f)$ is an $m$-multiple of a balanced 1-Hadamard code ${\mathscr H^{\rm 1b}_{2^k}}$ for some $k\geq 3$ and $m\geq 1$, or $\delta_0\nmid f$ and $|\mathscr S(f)|=3$.

  A symmetric statement holds for the $\delta_0$-affine kernel by switching $0$ and $1$ everywhere. 
\end{theorem}


Due to the existence of newly discovered non-trivial tractable classes  corresponding to the   Hadamard code in the \#EO problem without 
assuming {ARS}, 
the complexity classification for the complex-valued Holant problem becomes significantly more challenging.
It is hard to predict how many more tractable classes with remarkable properties are yet to be discovered.
The results in this paper are a small step towards an ambitious goal. 
Very recently, independent of this work, Meng, Wang and Xia discover new tractable classes for the {\rm \#EO} problem \cite{meng2024ptimealgorithmstypicaleo}. Subsequently, they prove an $\mathrm{FP^{NP}/\#P}$ dichotomy for $\rm{\# EO}$ problems \cite{meng2025fpnp}.
In addition, we note that the presented result is not the first time that a tractable class for counting problems has been found to be related to coding theory. 
Even in the classification of real-valued Holant problems, a tractable family called local affine~\cite{Cai-Lu-Xia-holant-c} was already discovered using the well-known Hamming code. 
Are there other interesting tractable classes for counting problems that can be carved out with the help of coding theory?
In the other direction, it is also worth exploring whether these new tractability results for counting problems can be applied back to coding theory.
We leave the questions in both directions for further study. 



The paper is organized as follows. In Section~\ref{sec:pre}, we introduce the necessary definitions.
In Section~\ref{sec:algorithm}, we present a chain reaction algorithm which establishes the tractability result (Theorem~\ref{thm:tractablity}) for  $\delta_1$-affine and $\delta_0$-affine signatures.
In Section~\ref{sec: characterization}, we prove the characterization theorem (Theorem~\ref{thm:characterization}) for $\delta_1$-affine and $\delta_0$-affine kernels. 
In Section~\ref{sec:hardness}, we prove the hardness result (Theorem~\ref{thm:hard}). In fact, the proof of the hardness result critically relies on the characterization of $\delta_1$-affine and $\delta_0$-affine kernels.

\section{Preliminaries}\label{sec:pre}
    \subsection{Definitions and notations}
    
A 0-1 valued signature is determined by its support $\mathscr{S}(f)$. 
In this paper, we use $f$ and $\mathscr{S}(f)$ interchangeably.
We can view $\mathscr{S}(f)$ as a binary code of length $\text{\sc arity}(f)$ where each $\alpha\in \mathscr{S}(f)$ is a codeword and the size of this code is $|\mathscr{S}(f)|$.
We say $\mathscr{S}(f)$ (or equivalently $f$) has a constant Hamming weight $\ell$ if $\rm wt(\alpha)=\ell$ for every $\alpha\in\mathscr{S}(f)$.
For simplicity, we say $f$ is $\ell$-weighted or constant-weighted in general.
The following fact can be easily checked: any factor of a constant-weighted signature is still constant-weighted.
If a vertex $v$ in a graph is labeled by a signature $f=g\otimes h$, we can replace the vertex $v$ 
by two vertices $v_1$ and $v_2$ and label $v_1$ with the factor $g$ and $v_2$ with $h$, respectively.
The incident edges of $v$ become incident edges of $v_1$ and $v_2$ respectively,
according to the partition of variables of $f$ in the tensor product of $g$ and $h$. This does not change the value of the instance.
    
      We represent $f$ or $\mathscr{S}(f)$ by a matrix where each row corresponds to a vector in $\mathscr{S}(f)$.
    All matrices that are equal up to permutations of rows and columns represent the same signature.
    The columns indexing variables are ordered from smallest to largest index, and are omitted for convenience when the context is clear.
    For example, the signature $f_2(x_1,x_2,x_3,x_4)=\{1100, 1010, 1001\}$ can be written as
    $f_2=\left[\begin{smallmatrix}
         1 &1 & 0& 0\\
         1 &0 & 1& 0\\
         1 &0 & 0& 1
    \end{smallmatrix}\right]$.

    We use $\#_i^0(f)$ and $\#_i^1(f)$ to denote the number of $0$s and $1$s, respectively, in the $i$-th column of $\mathscr{S}(f)$, for every $i\in I(f)$. 
     The complement of a signature $f$, denoted by $\overline{f}$, is the signature represented by the matrix obtained by flipping $0$ and $1$ at each position of the matrix representing $f$, i.e., $\mathscr{S}(\overline{f})=\{\alpha\in\mathbb{Z}_2^n \mid \overline{\alpha}\in \mathscr{S}(f)\}$.
     Let $\widehat{f}(x_1,x_2,\ldots,x_n)=x_1x_2\cdots x_n\oplus f(x_1,x_2,\ldots,x_n)$.
     Then, $\mathscr{S}(\widehat{f})=\{\vec{1}^n\}\Delta \mathscr{S}(f)$.
     Similarly, let $\widecheck{f}(x_1,x_2,\ldots,x_n)=(1\oplus x_1)(1\oplus x_2)\cdots(1\oplus x_n)\oplus f(x_1,x_2,\ldots,x_n)$, then $\mathscr{S}(\widecheck{f})=\{\vec{0}^n\}\Delta \mathscr{S}(f)$.

       We use $\neq_2$ to denote the binary signature $\{01,10\}$. Pinning the $i$-th variable of $f$ to $0$ gives the signature
     $f(x_1,\ldots,x_{i-1},0,x_{i+1},\ldots,x_n)$ of arity $n-1$, denoted by $f_{i}^{0}$. Extracting the $i$-th variable of $f$ to $0$ gives the signature $(x_i\oplus 1)f(x_1,x_2,\ldots,x_n)$ of arity $n$, denoted by 
     $f^{x_i=0}$. 
     Clearly $f^{x_i=0}=\delta_0\otimes f_i^0$. Similarly we define $f_{i}^{1}$ and $f^{x_i=1}$. For example,
    $$
    (f_2)_2^0=
    \begin{bNiceArray}{ccc}[first-row]
         x_1 & x_3 & x_4\\
         1 & 1& 0\\
         1 & 0& 1
    \end{bNiceArray}
    , 
    \quad \text{and} \quad
    (f_2)^{x_2=0}=
    \begin{bNiceArray}{cccc}[first-row]
        x_1&x_2&x_3&x_4\\
         1 & 0 &1& 0\\
         1 & 0 &0& 1
    \end{bNiceArray}.
    $$
    
       We use $f_{ij}^{ab}$ to denote the signature $(f_i^a)_j^b=(f_j^b)_i^a$, where $i,j\in I(f)$ and $a,b\in \mathbb{Z}_2$. The signature $f^{x_i\neq x_j}$ is defined by $f_{ij}^{01}+f_{ij}^{10}$, where we consider the signatures as real-valued and "+" denotes real addition.
    For a vector $\alpha\in\mathbb{Z}_2^n$,
    we construct a vector $\vec{\alpha}^m\in \mathbb{Z}_2^{nm}$, called an $m$-multiple of $\alpha$, by taking the concatenation of $m$ copies of $\alpha$. 
    For a signature $f$, we define its \emph{$m$-multiple}, denoted by $f_{\times m}$, to be the signature $\mathscr{S}(f_{\times m})=\{\vec{\alpha}^m\mid \alpha \in \mathscr{S}(f)\}$. 
 For example, the $2$-multiple of $f_2$ is
    $$
    (f_2)_{\times 2}=
    \left[
    \begin{array}{cccccccc}
         1 &1 & 0& 0 &1 &1 & 0& 0\\
         1 &0 & 1& 0 &1 &0 & 1& 0\\
         1 &0 & 0& 1 &1 &0 & 0& 1
    \end{array}
    \right].
    $$

    \subsection{Affine signatures}
        In this section, we give some basic properties of affine signatures. From the definition, one can verify that affine signatures are closed under pinning, extracting, tensor product, and factoring. Moreover, $f\in \mathscr{A}$ if and only if $f_{\times m} \in \mathscr{A}$, for any $m\in \mathbb{Z}_+$.

\begin{lemma}\label{constant-weighted}
          If a signature $f\in \mathscr{A}$ is constant-weighted and $\delta_0,\delta_1\nmid f$, then $\#_i^0(f)=\#_i^1(f)$ for any $i\in I(f)$. Moreover, $f$ is an \rm{EO} signature.
     \end{lemma}

     \begin{proof}
         Suppose $f=f(x_1,x_2,\ldots,x_n)$ and $x_1,x_2,\ldots,x_k$ are free variables. 
         We first claim $1\leq k<n$. Otherwise, if $k=n$, then $f$ is not constant-weighted since $\vec{0}^n\in \mathscr{S}(f)$ and $\vec{1}^n\in \mathscr{S}(f)$. It is clear that for every free variable $x_i$ we have $\#_i^0(f)=\#_i^1(f)$. For every $k<i\leq n$, write $x_i=\lambda_0+\lambda_1 x_1 + \lambda_2 x_2 +\cdots+\lambda_k x_k$, where $\lambda_j\in\mathbb{Z}_2$ for $0\leq j \leq k$. We have $\lambda_1,\lambda_2,\ldots,\lambda_k$ are not all zeros since $\delta_0 \nmid f$ and $\delta_1\nmid f$. We may assume $\lambda_1\neq 0$. Flipping the value of $x_1$ in the representation of $x_i$, the value of $x_i$ also flips. Thus, $\#_1^0(f)=\#_1^1(f)$ implies $\#_i^0(f)=\#_i^1(f)$, since there is a one-to-one correspondence between the positions where $x_i=0$ and $x_i=1$. Therefore, if we count by columns, the number of $0$s and $1$s are equal in $\mathscr{S}(f)$. Because $f$ is constant-weighted, we can count by rows and deduce that $f$ is an EO signature.
     \end{proof}

     \begin{definition}[Pairwise opposite]
        Let $f$ be a signature of arity $2n$. We say $f$ is pairwise opposite if we can partition the $2n$ variables into $n$ pairs such that on $\mathscr{S}(f)$, two variables of each pair always take opposite values.
    \end{definition}
    
    \begin{lemma}[Lemma 5.7 in \cite{cai2020beyond}]\label{pairwise-opposite}
        If $f$ is an affine \rm{EO} signature, then $f$ is pairwise opposite.
    \end{lemma}
    
    \begin{proof}
        Here we give a simplified inductive proof. If $n=1$, this lemma is trivially true since $f=\delta_1\otimes\delta_0$. Now assume the lemma holds for affine EO signatures of arity less than $2n$ ($n>1$). Suppose $f=\delta_1^{\otimes m_1}\otimes\delta_0^{\otimes m_0}\otimes g$, where $m_1,m_0\geq 0$ and $\delta_1,\delta_0\nmid g$. Then $g$ is a constant-weighted affine signature. By Lemma~\ref{constant-weighted}, we know $g$ is an EO signature and $\#_i^0(g)=\#_i^1(g)$ for every $i\in I(g)$. Then $m_0=m_1$ since $f$ is EO. If $m_0=m_1>0$, then $g$ is pairwise opposite by induction. It follows that $f$ is also pairwise opposite. If $m_1=m_0=0$, suppose $|\mathscr{S}(g)|=2t$, then $\#_i^1(g)=\#_i^0(g)=t$ for every $i\in I(g)$. Consider $g^{x_1=0}\in \mathscr{A}$, it is a constant-weighted signature of support size $t$. Write $g^{x_1=0}=\delta_1^{\otimes m_1'}\otimes\delta_0^{\otimes m_0'}\otimes g'$, where $m_1',m_0'\geq 0$ and $\delta_1,\delta_0\nmid g'$. Again by Lemma~\ref{constant-weighted} we have $m_1'=m_0'>0$. We may assume $x_2$ is a variable taking constant 1 in $\mathscr{S}(g^{x_1=0})$. Then $(x_1,x_2)$ is an opposite pair in $\mathscr{S}(g)$. Consider $g^{x_1\neq x_2}\in \mathscr{A}$, it is an EO signature of arity $2n-2$. Thus, by induction $g^{x_1\neq x_2}$ is pairwise opposite. It follows that $f=g$ is pairwise opposite.
    \end{proof}
    
    \section{A chain reaction algorithm}\label{sec:algorithm}

In this section, we introduce new tractable classes of signatures: $\delta_1$-affine and $\delta_0$-affine signatures, denoted by $\mathscr{D}_1$ and $\mathscr{D}_0$ respectively.
We show that the $\#$EO problem specified by a set of signatures in $\mathscr{D}_1\cup \mathscr A$ or $\mathscr{D}_0\cup \mathscr A$ is tractable.
The tractability is established via a chain reaction algorithm.
Given an instance $\Omega$ of $\#$EO$(\mathscr{D}_1\cup \mathscr A)$, we may assume that there is at least one vertex labeled by a $\delta_1$-affine signature $f$. Otherwise, all signatures in the instance are affine, which is called an affine instance and can be solved by the Gaussian elimination algorithm.

Consider the factor $\delta_1$ of $f$. 
It is connected to another variable of some signature (which may be $f$ itself) using $\neq_2$. 
Then, the connected variable is pinned to $0$.
A key property that ensures our algorithm works is that after connecting $\delta_1$ to another variable---whether it is a variable of $f$ itself, an affine signature,
or another $\delta_1$-affine signature---either the resulting signatures are all affine signatures,
or we can realize another $\delta_1$ signature.
In the latter case, we call it a propagation step.
After a propagation step, the newly realized $\delta_1$ can be used to pin another variable to $0$.
Thus, we get a chain reaction that terminates when an affine instance is reached.
Note that after each propagation step, the number of edges (total variables) in the instance is reduced by $2$.
Therefore, the chain reaction terminates after at most $|E|/2$ steps, where $E$ is the edge set of the underlying graph in $\Omega$.
Intuitively, the factor $\delta_1$ of a $\delta_1$-affine signature
plays the role of a neutron in a nuclear chain reaction:
it initializes the reaction and generates a new $\delta_1$ after propagation, which allows the chain reaction to continue.
Eventually, the chain reaction terminates when a stable state (an affine instance) is reached.



       \begin{lemma}\label{tractable_lemma}
        Suppose that $f=\delta_1\otimes h$ is $\delta_1$-affine and $g\in \mathscr{D}_1\cup\mathscr A$. 
        For any variable $x_i$ of $g$, by connecting it with $\delta_1$ of $f$ using $\neq_2$,
        we can obtain either a zero signature or a signature $h \otimes \delta_1 \otimes g_{ij}^{01}$, where $g_{ij}^{01}\in \mathscr{D}_1\cup\mathscr A$ for some $j\in I(g)\backslash\{i\}$.
    \end{lemma}
    \begin{proof}
        If $g\in \mathscr{A}$, then by Lemma~\ref{pairwise-opposite}, $g$ is pairwise opposite. Let $x_j$ be the variable opposite to $x_i$. When pinning $x_i$ to 0, $x_j$ becomes a $\delta_1$ factor, i.e., $g_i^0=\delta_1(x_j)\otimes g_{ij}^{01}$, where $g_{ij}^{01}\in\mathscr{A}$. If $g\in \mathscr{D}_1\backslash\mathscr A$, write $g=\delta_1(x_j)\otimes g'$. If we connect the $\delta_1$ factor of $f$ with $x_j$ using $\neq_2$, we obtain a zero signature. Now assume $x_i$ is not a variable taking constant 1 of $g$. Then $(g')_i^0\in\mathscr A\cup \mathscr D_1$ since $g\in \mathscr D_1$. In this case, we obtain $h\otimes \delta_1(x_j)\otimes g_{ij}^{01}$, where $g_{ij}^{01}=(g')_i^0\in \mathscr A\cup \mathscr D_1$.
    \end{proof}

   \begin{theorem}
        {\rm \#EO($\mathscr{A} \cup \mathscr{D}_1$)} and {\rm \#EO($\mathscr{A} \cup \mathscr{D}_0$)} are tractable.
    \end{theorem}

    \begin{proof}
        We only prove \#EO($\mathscr{A} \cup \mathscr{D}_1$) is tractable, as the other case follows by symmetry.

Consider an instance $\Omega$ of \#EO($\mathscr{A} \cup \mathscr{D}_1$). 
If all signatures in $\Omega$ are affine, then we are done by the Gaussian elimination algorithm. Otherwise, there exists a vertex $u$ in $\Omega$ that is labeled by a $\delta_1$-affine signature $f=\delta_1\otimes h$. 
Without loss of generality, we name the variable associated with $\delta_1$ as $x_1$. 
Suppose that in $\Omega$, $x_1$ is connected to another variable $x_i$ of some signature $g$ labeling vertex $v$ using $\neq_2$.

\begin{itemize}
    \item If $u=v$, i.e., $f$ and $g$ are labeled on the same vertex in $\Omega$, then by definition the resulting signature $f^{x_1\neq x_i}=f^{10}_{1i}=h^0_i$ obtained from $f$ by connecting variables $x_1$ and $x_i$ is in $\mathscr D_1\cup \mathscr A$.   
    The instance is thus reduced to a smaller instance with two fewer edges where $f$ is replaced by $f^{x_1\neq x_i}$ of smaller arity. 
    Since $f^{x_1\neq x_i}\in \mathscr D_1\cup \mathscr A$, the reduced instance remains an instance of \#EO($\mathscr{A} \cup \mathscr{D}_1$).
        \item Otherwise, $f$ and $g$ are labeled on different vertices. By Lemma~\ref{tractable_lemma}, connecting $\delta_1$ of $f$ with variable $x_i$ of $g$ yields a signature $h \otimes \delta_1\otimes g'$ where $g'=g_{ij}^{01}\in \mathscr{D}_1\cup\mathscr A$ for some $j\in I(g)\backslash\{i\}$.
        By renaming variables and decomposing $h \otimes \delta_1\otimes g'$ into two parts $f=h \otimes \delta_1$ and $g'$, 
        we can reduce the original instance to a smaller instance with two fewer edges where the edges incident to $f$ are changed and $g$ is replaced by $g'$ of smaller arity. 
        Since $g'\in \mathscr D_1\cup \mathscr A$, the reduced instance remains an instance of \#EO($\mathscr{A} \cup \mathscr{D}_1$).
\end{itemize}
In both cases, the instance is reduced to a smaller instance of \#EO($\mathscr{A}\cup \mathscr{D}_1$) with two fewer edegs, thus can be solved recursively and the number of recursions is bounded by $|E|/2$ where $E$ is the edge set of the underlying graph of $\Omega$. 
  \end{proof}

\section{Characterization of $\delta_1$-affine and $\delta_0$-affine kernels}\label{sec: characterization}

Since the classes $\mathscr D_1$ and $\mathscr D_0$ are defined inductively, obtaining explicit expressions for them that are similar to other known tractable classes for counting problems, such as affine signatures, is challenging.
The tractability of $\mathscr D_1$ or $\mathscr D_0$ is achieved by eventually reducing the problem to one where all signatures are affine signatures.
Therefore, we focus on $\delta_1$-affine (or $\delta_0$-affine) signatures which yield an affine signature when any variable is pinned to $0$ (or $1$, respectively).
In other words, these $\delta_1$-affine (or $\delta_0$-affine) signatures reside at the first level of the inductive hierarchy.
Below, we give an equivalent definition for the $\delta_1$-affine kernel. Note that in Definition~\ref{delta1-affine-kernel} we require $h\notin \mathscr{A}$ and in the following lemma we require $\delta_0\nmid h$.

    \begin{lemma}
        An {\rm EO} signature $f\in K(\mathscr{D}_1)$ if and only if it satisfies the following two properties: 
    \begin{enumerate}
        \item $f=\delta_1^{\otimes m}\otimes h$ for some $m \in \mathbb{Z_+}$ where $h$ is a signature of positive arity, $\delta_1, \delta_0 \nmid h$; 
        \item 
        for every $i\in I(h)$, $h_i^0\in \mathscr{A}$.
    \end{enumerate}
     \end{lemma}

     \begin{proof}
         First assume $f$ satisfies the two properties in this lemma. Since $f$ is an EO signature and $m\geq1$, we have that $h$ is constant-weighted and not EO. If $h\in\mathscr{A}$, then $h$ would be an EO signature by Lemma~\ref{constant-weighted}, which is a contradiction. Therefore, $h\notin \mathscr{A}$. Thus, $f\in K(\mathscr D_1)$ by Definition~\ref{delta1-affine-kernel}.

         Conversely, assume $f\in K(\mathscr D_1)$ by Definition~\ref{delta1-affine-kernel}. 
        For a contradiction, suppose that 
        $\delta_0\mid h$. 
        Then there exists a variable $x_i$ of $h$ taking constant 0. 
        Thus, $h_i^0\in \mathscr{A}$ by the second property. It follows that $h=\delta_0\otimes h_i^0\in \mathscr{A}$, contradicting the first property. Therefore, $\delta_0\nmid h$.
     \end{proof}

      \begin{remark}
        The signature $f_2$ in Example~\ref{example-2} is a $\delta_1$-affine kernel. Any EO signature of support size 3 with at least one $\delta_1$ factor and no $\delta_0$ factor is trivially a $\delta_1$-affine kernel, since any signature of support size 1 or 2 must be affine. We say $f\in K(\mathscr{D}_1)$ (or $K(\mathscr{D}_0)$) is \emph{non-trivial} if $|\mathscr{S}(f)|>3$.
    \end{remark}

Below, we prove our main characterization theorem for non-trivial $\delta_1$-affine and $\delta_0$-affine kernels.
\begin{theorem}\label{main-theorem}
     An {\rm EO} signature $f$ is a non-trivial $\delta_1$-affine (or $\delta_0$-affine) kernel if and only if $\mathscr S(f)$ is an $m$-multiple of a balanced 1-Hadamard (or balanced 0-Hadamard respectively) code of size $2^k-1$ for some $k\geq 3$ and $m\geq 1$. 
\end{theorem}

We give a proof outline for our main theorem. 

In section~\ref{subsec: support size of kernels}, we prove $\delta_1$-affine kernels admit certain arity and support size by carefully counting the number of 0s and 1s in each column of its support. 
More specifically, we prove that a non-trivial $\delta_1$-affine kernel $f$ must have support size $2^k-1$ and arity $m2^k$, for some $m,k\in \mathbb{Z}_+,k\geq 3$ (Lemma~\ref{size-theorem}). 

In section~\ref{subsec: basic kernel and m-multiple}, we define the basic $\delta_1$-affine kernel (Definition~\ref{basic-delta_1-affine-kernel}) of support size $2^k-1$ and arity $2^k$. We show that every non-trivial $\delta_1$-affine kernel is a multiple of such basic kernels (Lemma~\ref{multiple_kernel}).

In section~\ref{subsec: balanced Hadamard codes characterize basic kernels}, we characterize the basic $\delta_1$-affine kernel by the balanced 1-Hadamard code $\mathscr H_{2^k}^{1b}$. We define a special affine signature called butterfly (Definition~\ref{butterfly}), which serves as a bridge connecting basic kernels and balanced Hadamard codes. It is the ``maximal" affine signature of $k$ free variables, requiring that no two different variables take identical values. We show that $\mathscr H_{2^k}^{1}$ (or $\mathscr H_{2^k}^{0}$) is just half of butterfly (Lemma~\ref{Hadamard_is_butterfly}). We also define left and right wings of butterfly (Definition~\ref{wings-of-butterfly}) by removing the all-0 (or all-1) row of half of butterfly. Since $\mathscr H_{2^k}^{0b}$ (or $\mathscr H_{2^k}^{1b}$) is obtained by removing the all-0 (or all-1) codeword of $\mathscr H_{2^k}^{1}$, wings of butterfly coincide with $\mathscr H_{2^k}^{0b}$ and $\mathscr H_{2^k}^{1b}$ (Lemma~\ref{Hadamard_is_butterfly}). Finally, we prove that basic kernels are precisely wings of butterfly (Lemma~\ref{butterfly_is_kernel}, Lemma~\ref{uniqueness}). Therefore, basic kernels, balanced Hadamard codes and wings of butterfly are precisely the same thing.

    \subsection{Support size  and arity of $\delta_1$-affine kernels}\label{subsec: support size of kernels}

   \begin{lemma}\label{extract-1-kernel}
        Suppose $f\in K(\mathscr{D}_1)$. 
        If $\delta_0 \nmid f^{x_i=1}$ for some $i\in I(f)$, then $f^{x_i=1}\in K(\mathscr{D}_1)$.
    \end{lemma}
    
    \begin{proof}
    Suppose $f=\delta_1^{\otimes m}\otimes h$, where $m\in \mathbb{Z}_+, \delta_1,\delta_0\nmid h$, and $h_j^0\in \mathscr{A}$ for every $j\in I(h)$.
    If $i\notin I(h)$, then $f^{x_i=1}=f\in K(\mathscr{D}_1)$ and we are done.
    
    Otherwise $i\in I(h)$. Then, $f^{x_i=1}=\delta_1^{\otimes m}\otimes h^{x_i=1}$.
    Since $\delta_0 \nmid f^{x_i=1}$, we have $\delta_0 \nmid h^{x_i=1}$. 
    Clearly, $\delta_1 \mid h^{x_i=1}$ since variable $x_i$ of $h$ is fixed to take value 1. 
    Factoring out all possible $\delta_1$ factors in $h^{x_i=1}$, we may write $h^{x_i=1}=\delta_1^{\otimes m_0}\otimes h_0$, where $\delta_1\nmid h_0$. Then $f^{x_i=1}=\delta_1^{\otimes (m+m_0)}\otimes h_0$, where $\delta_1,\delta_0\nmid h_0$. We claim $h_0$ has positive arity. Otherwise, $f^{x_i=1}=\delta_1^{\otimes (m+m_0)}$ which is not an EO signature, contradiction. 
    Moreover, for every $j \in I(h_0)$, we have $\delta_1^{\otimes m_0}\otimes(h_0)_j^0=(h^{x_i=1})_j^0=(h_j^0)^{x_i=1}\in \mathscr{A}$ since $h_j^0\in \mathscr{A}$. It follows that $(h_0)_j^0\in \mathscr{A}$ for every $j\in I(h_0)$. Therefore, $f^{x_i=1}\in K(\mathscr{D}_1)$.
    \end{proof}

     \begin{lemma}\label{size-theorem}
     Let $f=\delta_1^{\otimes m}\otimes h\in K(\mathscr{D}_1)$, where $m\in\mathbb{Z}_+,\delta_1,\delta_0\nmid h$. If $\lvert \mathscr{S}(f)\rvert>3$, then there exists $k\in\mathbb{Z},k\geq2$ such that $|\mathscr{S}(f)|=2^{k+1}-1$ and $\text{\sc arity}(f)=m2^{k+1}$. Moreover, we have $\delta_1\nmid h_i^0,\delta_0\nmid h_i^1$ and $\#_i^0(h)=2^k$ for every $i\in I(h).$
    \end{lemma}
    
    \begin{proof}
    Suppose $h=h(x_1,x_2,...,x_n)$ has arity $n$ and is $t$-weighted. 
    Let $s=n-t$. 
    Since $f$ is an EO signature, we have $s=t+m$. Then $s>t$ since $m\in \mathbb{Z}_+$. 
    For every $i\in I(h)$ we have $h_i^0\in \mathscr{A}$, so $\#_i^0(h)=2^{k_i}$ for some $k_i\in \mathbb{Z}$. 
    Without loss of generality, we may assume the first column of $\mathscr{S}(h)$ contains the most $0$s, i.e., $k_1$ is the largest among all $k_i,1\leq i\leq n$. 
    Let $k=k_1$ and $l=\#_1^1(h)$. 
    We have $l>0$ since $\delta_0\nmid h$, and $|\mathscr{S}(h)|=2^k+l$. 
    Suppose $h_1^0=\delta_0^{\otimes q_0}\otimes\delta_1^{\otimes q_1}\otimes h'$ and $\text{\sc arity}(h')=p$, where $q_0,q_1,p\geq 0$, and $\delta_0,\delta_1\nmid h'$.
    An illustration of $\mathscr{S}(h)$ is shown in Figure~\ref{fig: support of h}, where yellow parts are all $0$s, blue parts are all $1$s, the green part is $h'$, and with $h_1^0$ (resp. $h_1^1$) enclosed by the red (resp. pink) box.

\begin{figure}[htbp]
    \centering
    \NiceMatrixOptions{xdots={line-style = <->}}
    $$\mathscr{S}(h)=
\begin{NiceArray}{@{\hspace{2em}}|c|cw{c}{1cm}cc|cc|cc|}[first-row,first-col,margin]
\CodeBefore
\cellcolor[HTML]{FFFF88}{1-1,2-1,3-1,4-1}
\cellcolor[HTML]{96e5fc}{5-1,6-1,7-1}
\cellcolor[HTML]{80fe95}{1-2,2-2,3-2,4-2}
\Body
  & x_1 & x_2 &\cdots\\
\hline
\Vdotsfor{4}^{2^k} & 0 & 0 & \Block[fill=[HTML]{80fe95}]{4-3}{h'} & & & \Block[fill=[HTML]{FFFF88}]{4-2}{0\text{s}} & & \Block[fill=[HTML]{96e5fc}]{4-2}{1\text{s}} & \\
& 0 & 0 & & & & & & & \\
& 0 & 1 & & & & & & & \\
& 0 & 1 & & & & & & & \\\hline
\Vdotsfor{3}^{l} & 1 & 0 &  \Block{3-3}{ } & & &\Block[fill=[HTML]{96e5fc}]{3-2}{1\text{s}} & & & \\
& 1 & 0 &  & & & & & & \\
& 1 & 1 &  & & & & & & \\\hline
\CodeAfter
\OverBrace[yshift=10pt]{1-2}{1-5}{p}
\OverBrace[shorten,yshift=3pt]{1-6}{1-7}{q_0}
\OverBrace[shorten,yshift=3pt]{1-8}{1-9}{q_1}
 \tikz[line width=0.7mm, pink]
 \draw (5-|2) -- (5-|last)
       (5-|2) -- (last-|2)
       (last-|2) -- (last-|last)
       (5-|last) -- (last-|last);
\tikz[line width=0.7mm, red]
 \draw (1-|2) -- (1-|last)
       (1-|2) -- (5-|2)
       (5-|2) -- (5-|last)
       (1-|last) -- (5-|last);
\tikz
\node[anchor=west, xshift=2mm, font=\normalsize] at ($(3-|last)!0.5!(3-|last)$) {$h_1^0$};
\tikz
\node[anchor=west, xshift=2mm, font=\normalsize] at ($(6-|last)!0.5!(7-|last)$) {$h_1^1$};
\end{NiceArray}$$
    \caption{An illustration of the support of $h$.}
    \label{fig: support of h}
\end{figure}

   Next we prove this lemma by five claims:
    
    \begin{claim}\label{support-size-1}
        $2^{k-1}<l<2^k$ with $k\ge 2$.
     \end{claim}
         
     Since $h$ is constant-weighted and $s>t$, the total number of $0$s outweighs that of $1$s in $\mathscr{S}(h)$. By the pigeonhole principle, the first column contains strictly more $0$s than $1$s, i.e., $2^k>l$. Thus, $2^{k}>\frac{1}{2}(2^k+l)=\frac{1}{2}\lvert \mathscr{S}(f)\rvert \geq2$. More precisely, $2^{k}\geq 4$ and $k\geq 2$.

     Next we show $l\ge 2^{k-1}.$
     First, we have $p>0$. Otherwise, $h_1^0=\delta_0^{\otimes q_0}\otimes\delta_1^{\otimes q_1}$, so the first $2^{k}$ rows would be identical in $\mathscr{S}(h)$ and $\mathscr{S}(f)$. This would imply $2^k=1$, contradicting $k\geq 2$. Moreover, if there exists some variable $x_i$ which is constant 0 in $\mathscr{S}(h_1^0)$, then $x_i$ is constant 1 in $\mathscr{S}(h_1^1)$ since the first column of $\mathscr{S}(h)$ contains the most $0$s.
     Since $f$ is an EO signature, $h_1^0$ is constant-weighted, and therefore so is $h'$. Clearly $h'\in \mathscr{A}$. 
     By Lemma~\ref{constant-weighted}, $h'$ is an EO signature and $$\#_i^0(h')=\#_i^1(h')=2^{k-1}$$
     for every $i\in I(h')$.
     Consequently, $q_0=s-1-p/2$ and $q_1=t-p/2$. 
     For every $i\in I(h')\subseteq I(h)$, we have $$2^{k_i}=\#_i^0(h)=\#_i^0(h_1^0)+\#_i^0(h_1^1)=\#_i^0(h')+\#_i^0(h_1^1)=2^{k-1}+\#_i^0(h_1^1).$$
     Since $k_i\leq k$, there are two possibilities for every $i\in I(h')$: 
     either $k_i=k$ and $\#_i^0(h_1^1)=2^{k-1}$, or $k_i=k-1$ and $\#_i^0(h_1^1)=0$. We claim that at least for one $i\in I(h')$, it is the former case. 
     Otherwise, $\#_i^0(h_1^1)=0$ for every $i\in I(h')$, meaning $x_i$ is constant 1 in $\mathscr{S}(h_1^1)$ for every $i\in I(h')$. 
     Counting the number of $1$s in the last row of $\mathscr{S}(h)$, we would have 
     $$\mathrm{wt}(h)\geq1+p+q_0=1+p+(s-1-p/2)=s+p/2>t,$$ contradicting wt$(h)=t$. Therefore, there exists some $i_0$ such that $\#_{i_0}^0(h_1^1)=2^{k-1}$. 
     Without loss of generality, we assume $i_0=2$.
     Then $l=\#_1^1(h)\geq \#_2^0(h_1^1)= 2^{k-1}$.

     Finally, we show $l\neq2^{k-1}$ and thus $l>2^{k-1}$. 
     Assume by contraction that $l=2^{k-1}$.
     For every $i\in I(h')$, $x_i$ is constant 0 or constant 1 in $\mathscr{S}(h_1^1)$ since $\#_i^0(h_1^1)=0$ or $\#_i^0(h_1^1)=2^{k-1}=l$ by the above argument. 
     Now we show the last $l$ rows in $\mathscr{S}(h)$ must be identical to get a contradiction. The only uncertain part is the last $q_1$ columns.
     Consider $h_2^0\in \mathscr{A}$, it has support size $2^k$. 
     For $x_i$ which is a variable from the last $q_1$ columns, it is constant 1 or half $0$s and half $1$s in $\mathscr{S}(h_2^0)$ by Lemma~\ref{constant-weighted}. Since $\#_2^0(h_1^0)=2^{k-1}$, the column corresponding to $x_i$ already contains $2^{k-1}$ many $1$s in $\mathscr{S}(h_2^0)$. 
     So $x_i$ is constant 1 or constant 0 in the last $2^{k-1}$ rows of $\mathscr{S}(h_2^0)$, which is $\mathscr{S}(h_1^1)$. Therefore, the last $2^{k-1}$ rows in $\mathscr{S}(h)$ are identical. This would imply $2^{k-1}=1$, contradicting $k\geq 2$.
    
      \begin{claim}\label{support-size-2}
     $\#_i^0(h)=2^k$ for every $i\in I(h')$.
    \end{claim}

    Since $\#_i^0(h)=\#_i^0(h_1^0)+\#_i^0(h_1^1)$ and $\#_i^0(h_1^0)=2^{k-1}$, we only need to show $\#_i^0(h_1^1)=2^{k-1}$ for every $i\in I(h')$.
    In Claim~\ref{support-size-1}, we have shown either $\#_i^0(h_1^1)=2^{k-1}$ or $\#_i^0(h_1^1)=0$. 
    If Claim~\ref{support-size-2} is not true, then there exists some $i_1\in I(h')$ such that $\#_{i_1}^0(h_1^1)=0$. 
    Without loss of generality, we assume $i_1=3$. Then $x_3$ is constant 1 in $\mathscr{S}(h_1^1)$.
   For $i\in I(h')$, since $\#_i^0(h_1^1)\leq 2^{k-1}<l$, $x_i$ is not constant $0$ in $\mathscr{S}(h^{x_3=1})$. Additionally, the last $q_0+q_1$ columns cannot be constant $0$s in $\mathscr{S}(h^{x_3=1})$. Thus, $\delta_0 \nmid f^{x_3=1}$. By Lemma~\ref{extract-1-kernel}, we have $f^{x_3=1}\in K(\mathscr{D}_1)$.
    Note that $\lvert\mathscr{S}(f^{x_3=1})\rvert=2^{k-1}+l$, where $2^{k-1}<l<2^k$. Let $l=2^{k-1}+l_1$, where $0<l_1<2^{k-1}$. Then $\lvert\mathscr{S}(f^{x_3=1})\rvert=2^{k}+l_1>3$. By Claim~\ref{support-size-1} we have $2^{k-1}<l_1<2^k$. This is a contradiction. Therefore, $\#_i^0(h)=2^k$ for every $i\in I(h')$.
    
    \begin{claim}\label{support-size-3}
    $l=2^k-1$, $\lvert\mathscr{S}(f)\rvert=2^{k+1}-1$.
    \end{claim}
    
    By Claim~\ref{support-size-1}, we have $|\mathscr{S}(f)|=2^k+l$, where $k,l\in\mathbb{Z},k\geq2,2^{k-1}<l<2^k$.

    In Claim~\ref{support-size-2}, we prove $\#_i^0(h_1^1)=2^{k-1}$ for every $i\in I(h')$. 
    Thus, for every $i\in I(h')$, $x_i$ is not constant $0$ in $\mathscr{S}(h_1^1)$ since $l>2^{k-1}$. No variable from the last $q_1$ columns takes constant 0 in $\mathscr{S}(h_1^1)$ either, since $2^{k-1}<l<2^k$ is not a power of 2. Therefore, $\delta_0 \nmid f^{x_1=1}$. By Lemma~\ref{extract-1-kernel}, $f^{x_1=1}\in K(\mathscr{D}_1)$. If $|\mathscr{S}(f^{x_1=1})|=l>3$, we apply Claim~\ref{support-size-1} and obtain $l=2^{k-1}+l_1$, where $k-1\geq2$ and $2^{k-2}<l_1<2^{k-1}$. Thus, $|\mathscr{S}(f)|=2^k+l=2^k+2^{k-1}+l_1$. Repeating this process, we finally obtain $\lvert\mathscr{S}(f)\rvert=2^{k}+l=2^k+2^{k-1}+l_1=...=2^k+2^{k-1}+...+2^2+r$, where $2^1<r<2^2$. Then $r=3,l=2^k-1$ and $\lvert\mathscr{S}(f)\rvert=2^{k+1}-1$.

      \begin{claim}\label{support-size-4}
    $\delta_1\nmid h_i^0,\delta_0\nmid h_i^1$, and $\#_i^0(h)=2^k$ for every $i\in I(h)$.
\end{claim}

    First we show $\delta_1\nmid h_1^0$, i.e., $q_1=0$.
    Otherwise, let $\delta_1(x_i)\mid h_1^0$ for some $i\in I(h)$.
    By Claim~\ref{support-size-3}, we know $l=2^k-1$ and $|\mathscr{S}(f)|=2^{k+1}-1$. 
    Since $\#_i^0(h_1^1)=\#_i^0(h)=2^{k_i}$ for some $k_i\in\mathbb{Z}$ and $2^{k_i}\leq l=2^{k}-1$, we have $\#_i^1(h_1^1)=(2^k-1)-\#_i^0(h_1^1)=2^k-1-2^{k_i}>0$. Let $l_1=\#_i^1(h_1^1)>0$. Then consider $f^{x_i=1}$, it is easy to verify it has no $\delta_0$ factor. By Lemma~\ref{extract-1-kernel}, $f^{x_i=1}\in K(\mathscr{D}_1)$. Since $|\mathscr{S}(f^{x_i=1})|=2^k+l_1>2^k\geq4$, we can apply Claim~\ref{support-size-3} and deduce $|\mathscr{S}(f^{x_i=1})|=2^k+l_1=2^k+2^k-1$. Therefore $l_1=2^k-1$. Then $x_i$ takes constant $1$ in $\mathscr{S}(h)$, contradicting $\delta_1\nmid h$. Thus, $\delta_1\nmid h_1^0$.

    Second, in Claim~\ref{support-size-2} we proved $\#_i^0(h_1^1)=2^{k-1}$ or $0$ for every $i\in I(h_1^1)$. It follows that $\delta_0\nmid h_1^1$.
    
    Next we show $\#_i^0(h)=2^k$ for every $i\in I(h)$. We have shown $\#_i^0(h)=2^{k}$ for every $i\in I(h')$ in Claim~\ref{support-size-2}. 
    Also, $\#_i^0(h)=2^k$ for $x_i$ taking constant $0$ in $\mathscr{S}(h^{x_1=0})$. Combining $\delta_1\nmid h_1^1$, we get $\#_i^0(h)=2^k$ for every $i\in I(h)$.
    
    Now every column in $\mathscr{S}(h)$ contains the same number of $0$s, then the argument for $\delta_1\nmid h_1^0$ and $\delta_0\nmid h_1^1$ applies to any $i\in I(h)$.

    \begin{claim}\label{support-size-5}
    $\text{\sc arity}(f) =m2^{k+1}$.
    \end{claim}

     By Claim~\ref{support-size-4}, $\#_i^0(h)=2^k$ for every $i\in I(h)$. We count the number of $0$s in $\mathscr{S}(h)$ by columns and get $n2^k$. Since $h$ is constant-weighted and $\lvert\mathscr{S}(h)\rvert=2^{k+1}-1$ by Claim~\ref{support-size-3}, we count the number of $0$s by rows and get $s(2^{k+1}-1)$. Recall that $s+t=n$ and $s-t=m$. Thus, $s(2^{k+1}-1)=n2^k=(2s-m)2^k$. Solving this equation we obtain $s=m2^k$, then $\text{\sc arity}(f)=2s=m 2^{k+1}$.
  \end{proof}
  
    \subsection{Basic kernel and $m$-multiple}\label{subsec: basic kernel and m-multiple}

      \begin{definition}[Basic $\delta_1$-affine kernel of order $k$]\label{basic-delta_1-affine-kernel}
        If $m=1$ in Lemma~\ref{size-theorem}, i.e., $f=\delta_1\otimes h\in K(\mathscr{D}_1), \delta_1,\delta_0\nmid h$, then we call $f$ a basic $\delta_1$-affine kernel of order $k$, where $f$ has arity $2^k$ and support size $2^k-1$ $(k\geq 3)$. For completeness, we define the basic $\delta_1$-affine kernel of order 1 by $f_1=\delta_1\otimes\delta_0$ and the basic $\delta_1$-affine kernel of order 2 by $f_2=\delta_1\otimes\{(001),(010),(100)\} $.
    \end{definition}

\begin{lemma}\label{multiple_kernel}
    An \rm{EO} signature of arity $m2^k$ is a non-trivial $\delta_1$-affine kernel if and only if it is an $m$-multiple of a non-trivial basic $\delta_1$-affine kernel of arity $2^k$.
\end{lemma}

\begin{proof}
     First, we show that an $m$-multiple of a non-trivial basic $\delta_1$-affine kernel is a non-trivial $\delta_1$-affine kernel.
     Assume $f=\delta_1\otimes g$ is a non-trivial basic $\delta_1$-affine kernel. Then $|\mathscr{S}(f_{\times m})|=|\mathscr S(f)|>3$. We want to show $f_{\times m}\in K(\mathscr D_1)$. Note that $f_{\times m}=\delta_1^{\otimes m}\otimes g_{\times m}$ and $\delta_1,\delta_0\nmid g_{\times m}$. Pinning $0$ to the $i$-th variable of one copy of $g$ in $g_{\times m}$, each of the remaining $m-1$ copies contributes a $\delta_0$ factor, i.e., $(g_{\times m})_i^0=\delta_0^{\otimes(m-1)}\otimes (g_i^0)_{\times m}$. Clearly $(g_{\times m})_i^0\in \mathscr A$ since $g_i^0\in\mathscr{A}$. It follows that $f_{\times m}\in K(\mathscr D_1)$.

     Conversely, assume $f=\delta_1^{\otimes m}\otimes h$ is a non-trivial $\delta_1$-affine kernel of order $k$, where $m\in\mathbb{Z}_+, \delta_1,\delta_0\nmid h$. We want to show there exists a non-trivial basic $\delta_1$-affine kernel $f_0$ of order $k$ such that $f=(f_0)_{\times m}$.
      By Lemma~\ref{size-theorem}, we have  $\lvert\mathscr{S}(f)\rvert=2^{k}-1$ and $\#_i^0(h)=2^{k-1}$ for every $i\in I(h)$. Fix any $i\in I(h)$, consider $f^{x_i=0}=\delta_1^{\otimes m}\otimes h^{x_i=0}$. It is an affine EO signature of support size $2^{k-1}$. By Lemma~\ref{size-theorem}, $\delta_1 \nmid h^{x_i=0}$. Factoring out $\delta_0$ in $h^{x_i=0}$, we suppose $h^{x_i=0}=\delta_0^{\otimes m_0} \otimes h'$, $\delta_1,\delta_0 \nmid h'$. At the same time, notice $h'$ is constant-weighted. Thus, $h'$ is an EO signature by Lemma~\ref{constant-weighted}.
    Since $f^{x_i=0}=\delta_1^{\otimes m} \otimes \delta_0^{\otimes m_0} \otimes h'$ is also EO, we have $m=m_0$. 
    
    On one hand, if $x_j$ is a variable taking constant $0$ in $\mathscr{S}(f^{x_i=0})$, then the values of $x_j$ and $x_i$ are identical in $\mathscr{S}(h)$, since the $j$-th column already contributes $2^{k-1}$ many $0$s in $\mathscr{S}(f^{x_i=0})$ and the remaining $2^{k-1}-1$ rows must be all $1$s. 
    On the other hand, if $x_j$ is a variable of $h'$, then $x_j$ and $x_i$ are not identical in $\mathscr{S}(h)$ since they are not identical in $\mathscr{S}(h^{x_i=0})$. Therefore, there are exactly $m_0=m$ copies of the $i$-th column in $\mathscr{S}(h)$. Since $i\in I(h)$ is arbitrary, we have $h=(h_0)_{\times m}$ for some signature $h_0$.
    Clearly $\delta_1,\delta_0\nmid h_0$ since $\delta_1,\delta_0\nmid h$.
    Notice for every $i\in I(h_0)$, $(h_0^{x_i=0})_{\times m}=((h_0)_{\times m})^{x_i=0}=h^{x_i=0}=\delta_0\otimes h_i^0\in \mathscr{A}$, so $h_0^{x_i=0}\in \mathscr{A}$ and then $(h_0)_i^0\in \mathscr{A}$. Therefore, $f_0=\delta_1\otimes h_0$ is a basic $\delta_1$-affine kernel and $f=(f_0)_{\times m}$. 
\end{proof}

    \subsection{Balanced Hadamard codes characterize basic kernels}\label{subsec: balanced Hadamard codes characterize basic kernels}

    We first define butterfly signatures and wings of a butterfly, which serve as a bridge in establishing the equivalence between balanced Hadamard codes and basic kernels.

      \begin{definition}[Butterfly]\label{butterfly}
      
    Let $k\geq 1$.
    An affine signature $g$ of arity $2^{k+1}$ is a butterfly and denoted by $B_k$, if it has exactly $k$ free (linearly independent) variables $x_1,x_2,...,x_k$ and any two different variables do not always take identical values in $\mathscr{S}(g)$.        
    \end{definition}

    The butterfly $B_k$ is unique up to a permutation of variables. It is precisely the affine signature of $k$ free variables with the maximum arity, assuming any two different variables are not identical.

    \begin{definition}[Wings of butterfly]\label{wings-of-butterfly}
    We represent the butterfly $B_k$ by a matrix in the following way.
    Let $(x_2,x_3,...,x_{k+1})$ be free variables of $B_k$, and $x_1$ be constant $0$. Moreover, let $x_m=\lambda_2x_2+\lambda_3x_3+...+\lambda_{k+1}x_{k+1}$ and $x_{m+2^k}=x_{m}+1$ for $1\leq m\leq 2^k$, where $\lambda_2,...,\lambda_{k+1}\in \Z_2$. Through a permutation of rows, we let the first row be $\vec{0}^{2^k}\vec{1}^{2^k}$. 
    
    Then, $B_k$ can be split into the left half and the right half, each consisting of $2^k$ columns.
    The left wing of $B_k$, denoted by $L_k$, is the signature obtained from the left half of $B_k$ by removing the first row $\vec{0}^{2^k}$.
    Symmetrically, we define the right wing of $B_k$ denoted by $R_k$.
    \end{definition}

        \begin{remark}
        Recall that $\widecheck{L_k}$ and $\widehat{R_k}$ are obtained from $L_k$ and $R_k$ by adding an all-0 or all-1 row respectively. 
        They are the left and right halves of $B_k$ respectively, which are affine.
        In fact, $\widecheck{L_k}$ is precisely the affine signature of $k$ free variables with maximum arity, assuming every variable is a linear combination of free variables and any two variables are not identical. 
        Meanwhile, $\widehat{R_k}$ is the complement of $\widecheck{L_k}$.
        The uniqueness proof in Lemma~\ref{uniqueness} essentially relies on the maximality of $L_k$ and $R_k$.
    \end{remark}

\begin{example}\label{butterfly_example}

    $$
    B_2=
    \begin{bNiceArray}{cccccccc}[first-row,margin]
        x_1&x_2&x_3&x_4&x_5&x_6&x_7&x_8\\
         0 &0 & 0& 0&1 & 1 & 1& 1\\
         0 &1 & 1& 0&1 & 0 & 0& 1\\
         0 &1 & 0& 1&1 & 0 & 1& 0\\
         0 &0 & 1& 1&1 & 1 & 0& 0\\
    \end{bNiceArray}
    $$
    is the butterfly of $2$ free variables $x_2,x_3$. It is determined by the following linear relations:$x_1=0x_2+0x_3,x_4=x_2+x_3,x_5=0x_2+0x_3+1,x_6=x_2+0x_3+1,x_7=0x_2+x_3+1,x_8=x_2+x_3+1$.
    
    The left and right wing of $B_2$ are:
    $$
    L_2=
    \begin{bNiceArray}{cccc}
       0&  1 &1& 0\\
       0&  1 &0& 1\\
       0&  0 &1& 1
    \end{bNiceArray}
    ,
    \quad
    R_2=
    \begin{bNiceArray}{cccc}
        1& 0 & 0& 1\\
        1& 0 & 1& 0\\
        1& 1 & 0& 0
    \end{bNiceArray}
    $$
   \end{example}

        \begin{lemma}\label{Hadamard_is_butterfly}
        As a signature, the balanced 1-Hadamard (or 0-Hadamard) code is the right (or left) wing of butterfly up to a permutation of variables. 
    \end{lemma}
\begin{proof}
        We only need to prove $\mathscr H_{2^k}^0$ is the left half of $B_k$. It is easy to check $\mathscr H_{2^k}^0$ is affine as a signature. By the definition of Hadamard matrix and Hadamard code, every two different rows of $\mathscr H_{2^k}^0$ take different values. By Sylvester's construction, $\mathscr H_{2^k}^0$ is symmetric as a matrix. Thus, every two different variables of $\mathscr H_{2^k}^0$ take different values in its support. Since $\mathscr H_{2^k}^0$ has arity $2^k$, support size $2^k$, and contains the all-0 vector in its support, it is precisely the affine signature of $k$ free variables with maximum arity, assuming every variable is a linear combination of free variables and any two variables are not identical.
        Therefore, $\mathscr H_{2^k}^{0b}=L_k$ is obtained by removing the all-0 row from $\mathscr H_{2^k}^0$.
    \end{proof}
  
    So far we have shown the equivalence of balanced Hadamard code and wings of butterfly. Next we focus on the relationship between basic kernels and wings of butterfly. For convenience, we only consider the basic $\delta_1$-affine kernel and the right wing of butterfly in the following. The other case is symmetric.

  \begin{lemma}\label{butterfly_is_kernel}
    Let $k\geq2$ and $R_k$ be the right wing of $B_k$, then $R_k$ is a basic $\delta_1$-affine kernel of order $k$.  
    \end{lemma}

    \begin{proof}
        Note that $R_k=\mathscr H_{2^k}^{1b}$ by Lemma~\ref{Hadamard_is_butterfly}.
        So $R_k$ is an EO signature, and every two different rows and columns of $R_k$ do not take identical or opposite values.
        Note that $R_k=\delta_1\otimes r_k$, $\delta_0,\delta_1\nmid r_k$, and $({r_k})_i^0=(\widehat{r_k})_i^0\in \mathscr{A}$, since $\widehat{r_k}\in \mathscr{A}$. Finally, since $R_k$ has arity $2^k$ and support size $2^k-1$, it is a basic $\delta_1$-affine kernel of order $k$.
    \end{proof}

    \begin{corollary}\label{cor: non-singularity}
        If the basic $\delta_1$-affine kernel $f_k$ is unique (up to a permutation of variables), then no two distinct rows or columns in $\mathscr{S}(f_k)$ can have identical or opposite entries.
    \end{corollary}
    \begin{proof}
        Note that this property holds for $R_k$.
    \end{proof}
        \begin{lemma}\label{long_lemma}
        Suppose $k\in \Z_+,k\geq2$. If the $k$-th basic $\delta_1$-affine kernel $f_{k}$ is unique up to a permutation of variables, then any $(k+1)$-th basic $\delta_1$-affine kernel $f=\delta_1\otimes h$ satisfies the following properties for every $l\in I(h)$:
    \begin{enumerate}
        \item 
        $f^{x_l=1}=(f_k)_{\times 2}$ up to a permutation of variables.
        \item 
        $f^{x_l=0}=B_{k}$ up to a permutation of variables.
        \item 
        Two variables $x_i$ and $x_j$ take opposite values in $\mathscr{S}(f^{x_l=0})$ if and only if they take identical values in $\mathscr{S}(f^{x_l=1})$.
    \end{enumerate}
    \end{lemma}

    \begin{proof}
        Let $f=\delta_1(x_0)\otimes h(x_1,x_2,...,x_{2^{k+1}-1})$. 
        We first prove the three claims for $k>2$.
        By Lemma~\ref{size-theorem},
        $\#_l^0(h)=2^k$ for every $l\in I(f)$. 
        Without loss of generality, we assume $l=1$ in the following argument. 

    \begin{claim}
         $f^{x_l=1}=(f_k)_{\times 2}$. 
    \end{claim}
    
    Up to a permutation of rows, we assume $x_1=0$ in the first $2^k$ rows of $\mathscr{S}(f)$, and $x_1=1$ in the last $2^k-1$ rows of $\mathscr{S}(f)$. 
    By Lemma~\ref{size-theorem}, $\delta_0\nmid h_1^1$. Then $\delta_0\nmid f^{x_1=1}$.
    By Lemma~\ref{extract-1-kernel}, $f^{x_1=1}\in K(\mathscr{D}_1)$. It has at least two $\delta_1$ factors $\delta_1(x_0)$ and $\delta_1(x_1)$. Let $f^{x_1=1}=\delta_1^{\otimes m}\otimes f'$ where $m\geq2$ and $\delta_1,\delta_0\nmid f'$.
    Since $|\mathscr{S}(f^{x_1=1})|=2^k-1>3$, we can apply Lemma~\ref{multiple_kernel} to get $f^{x_1=1}=(f_k)_{\times m}$. 
    Because $\text{\sc arity}(f_k)=2^k$ and $\text{\sc arity}(f^{x_1=1})=2^{k+1}$, we have $m=2$. Therefore, $f^{x_1=1}=(f_k)_{\times2}$, where $f_k$ is the unique basic $\delta_1$-affine kernel of order $k$.

    \begin{claim}
        $f^{x_1=0}=B_k$.
    
    \end{claim}

    Since $f^{x_1=0}$ is affine and has support size $2^k$ and arity $2^{k+1}$, we only need to show any two different variables do not take identical values in $\mathscr{S}(f^{x_1=0})$.
     First we show the $\delta_1$ and $\delta_0$ factor of $f^{x_1=0}$ is unique.
    By Lemma~\ref{size-theorem}, $\delta_1\nmid h_1^0$, so $\delta_1(x_0)$ is the unique $\delta_1$ factor of $f^{x_1=0}$. If $x_j$ takes constant $0$ in $\mathscr{S}(f^{x_1=0})$, then it takes constant $1$ in $\mathscr{S}(f^{x_1=1})$ since $\#_j^0(h)=2^k$. Because $f^{x_1=1}$ only has two $\delta_1$ factors $\delta_1(x_0)$ and $\delta_1(x_1)$, either $j=0$ or $j=1$. However, the former is impossible since $\delta_1(x_0)$ is a $\delta_1$ factor of $f$. Thus, we have $j=1$ and $\delta_0(x_1)$ is the only $\delta_0$ factor of $\mathscr{S}(f^{x_1=0})$.

    Next we show any two different variables $x_i,x_j(i<j)$ which do not take constant $0$ or $1$ in $\mathscr{S}(f^{x_1=0})$ do not take identical values in $\mathscr{S}(f^{x_1=0})$. We prove by contradiction. Let $f_k=\delta_1\otimes h_k$, and let $f_1^1(x_2,x_3,...,x_{2^k})$ and $f_1^1(x_{2^k+1},x_{2^k+2},...,x_{2^{k+1}-1})$ be two copies of $h_k$. 
    Moreover, assume $x_s$ and $x_{s+2^k-1}$ take identical values in $\mathscr{S}(f^{x_1=1})$ for every $2\leq s \leq 2^{k}$. 
    We first show $2\leq i\leq 2^k$ and $2^k+1\leq j\leq 2^{k+1}-1$ if $x_i$ and $x_j$ take identical values in $\mathscr{S}(f^{x_1=0})$. Otherwise, assume $2\leq i<j\leq 2^k$.
    Consider $h^{x_i=0}\in\mathscr{A}$. By Lemma~\ref{size-theorem}, its support size is $2^k$. 
    Compare the $i$-th and the $j$-th column of $\mathscr{S}(h^{x_i=0})$ - the first $2^{k-1}$ rows are both $0$s, the last $2^{k-1}$ rows of the $i$-th column are $0$s.
    By Lemma~\ref{constant-weighted}, the last $2^{k-1}$ rows of the $j$-th column must be all $0$s or all $1$s. In the former case, $x_i$ and $x_j$ take identical values in $\mathscr{S}(f_{k-1})$, as well as in $h_k$ because $2\leq i<j\leq 2^k$. 
    This contradicts Corollary~\ref{cor: non-singularity}.
    In the latter case, $\#_j^1(f^{x_1=1})\geq 2^{k-1}$, contradicting $\#_j^1(f^{x_1=1})=|\mathscr{S}(f^{x_1=1})|-\#_j^0(f^{x_1=1})=2^{k}-1-2^{k-1}=2^{k-1}-1$. 
    Therefore, $2\leq i\leq 2^k,2^k+1\leq j\leq2^{k+1}-1$. 
    Since $f^{x_1=0}$ is affine and EO, it is pairwise opposite by Lemma~\ref{pairwise-opposite}. Let $(x_i,x_i')$ be an opposite pair, then $(x_j,x_i')$ is also an opposite pair. 
    Since $x_i$ and $x_j$ are variables from different copies of $h_k$, there must be an opposite pair from the same copy of $h_k$. 
    Without loss of generality, assume $2\leq i < i'\leq 2^k$. 
    Compare the $i$-th and $i'$-th column of $\mathscr{S}(h^{x_i=0})$, the first $2^{k-1}$ rows are $0$s and $1$s respectively. The last $2^{k-1}$ rows of the $i$-th column are $0$s, so the $i'$-th column must be all $0$s or all $1$s. However, neither can happen by the same argument above. Therefore, we prove any two different variables do not take identical values in $\mathscr{S}(f^{x_1=0})$. It follows that $f^{x_1=0}=B_k$. 

    \begin{claim}
       $x_i$ and $x_j$ take opposite values in $\mathscr{S}(f^{x_1=0})$ iff they take identical values in $\mathscr{S}(f^{x_1=1})$. 
    \end{claim}
    
    If $x_i$ and $x_j$ both take constant $1$ in $\mathscr{S}(f^{x_1=1})$, they must take constant $0$ and $1$ respectively in $f^{x_i=0}$. If $x_i$ and $x_j$ take identical values other than constant $1$ in $\mathscr{S}(f^{x_1=1})$, consider $f^{x_i=0}\in \mathscr{A}$. 
    The last $2^{k-1}$ rows of $\mathscr{S}(f^{x_i=0})$ of the $i$-th and the $j$-th column are constant $0$s. The first $2^{k-1}$ rows of the $i$-th column is constant $0$. By Lemma~\ref{constant-weighted}, the $j$-th column must be all $1$s or all $0$s. However, the latter cannot happen, otherwise $x_j$ and $x_i$ take identical values in $\mathscr{S}(f^{x_1=0})$. Thus, the $j$-th column of the first $2^{k-1}$ rows in $\mathscr{S}(f^{x_i=0})$ must be all $1$s. Therefore, $x_i$ and $x_j$ take opposite values in $\mathscr{S}(f^{x_i=0})$. Conversely, we also have $x_i$ and $x_j$ take identical values in $\mathscr{S}(f^{x_1=1})$ if they take opposite values in $\mathscr{S}(f^{x_1=0})$.
    
    \vspace{1em}
    Finally, we discuss the case $k=2$. Let $f=\delta_1(x_0)\otimes h(x_2,...,x_7)$ be any basic $\delta_1$-affine kernel of order $3$. The above proof only fails in the first part when applying Lemma~\ref{multiple_kernel}. So we only need to show $f^{x_l=1}=(f_2)_{\times 2}$ for every $l\in I(f)$. By Lemma~\ref{size-theorem}, $\#_l^0(h)=4$ for every $l\in I(h)$. Without loss of generality, assume $l=1$, and $x_1=0$ in the first 4 rows of $\mathscr{S}(f)$, $x_1=1$ in the last 3 rows of $\mathscr{S}(f)$. We still have $f^{x_1=0}=B_2$ by the same argument above, then $\#_i^0(h_1^0)=2$ for every $i\in I(h_1^0)=\{2,3,4,5,6,7\}$. Then $\#_i^0(h_1^1)=4-2=2$. Thus, the six columns in $\mathscr{S}(h_1^1)$ must come from $\{(001)^T,(010)^T,(100)^T\}$. Note that none of them can appear more than 2 times. Otherwise, suppose there are 3 columns taking $(001)^T$, then there are five $1$s in the last row of $\mathscr{S}(f)$, contradicting that $f$ is an EO signature. However, there are six of them in total, so each must appear 2 times. Therefore, $h_1^1=\{(001),(010),(100)\}_{\times 2}$ and $f^{x_1=1}=(f_2)_{\times 2}$.
    \end{proof}

        \begin{lemma}\label{uniqueness}
        The basic $\delta_1$-affine kernel of order $k$, denoted by $f_k$, is precisely $R_k$ up to a permutation of variables.
    \end{lemma}
    
    \begin{proof}
    For $k=1$, we have $f_1=\delta_1\otimes\delta_0$ and this lemma is clear. For $k=2$, $f_2=\delta_1\otimes \{(001),(010),(100)\}$ and this lemma is true by Example~\ref{butterfly_example}. Now assume this lemma holds for positive integers less than or equal to $k(k\geq 2)$. By Lemma~\ref{long_lemma}, we know any $(k+1)$-th $\delta_1$-affine kernel $f_{k+1}=\delta_1\otimes h_{k+1}$ satisfies the following three properties:
    
\NiceMatrixOptions{xdots={line-style = standard}}
$$f_{k+1}=\begin{NiceArray}{*{8}c}[hvlines]
\Block{1-4}{} 0& 0& \cdots& 0& \Block{1-4}{} 1& 1& \cdots& 1\\
\Block{4-4}{\overline{f_k}} & & & & \Block{4-4}{f_k} & & &\\
& & & & & & &\\
& & & & & & &\\
& & & & & & &\\
\Block{4-4}{f'_k} & & & & \Block{4-4}{f'_k} & & &\\
& & & & & & &\\
& & & & & & &\\
& & & & & & &\\
\end{NiceArray}$$

\begin{enumerate}
    \item 
    Extracting any variable of $h_{k+1}$ to $0$, up to a permutation of rows, the first $2^k$ rows of $f_{k+1}$ is butterfly $B_k$.
    \item 
    Extracting any variable of $h_{k+1}$ to $1$, up to a permutation of rows, the last $2^k-1$ rows is $(f_k)_{\times 2}$. However, when we write the upper butterfly into the specific form shown in the above picture, we fix the order of columns. Then we denote the last $2^k-1$ rows by two $f'_k$, where $f'_k=f_k$ up to a permutation of rows and columns.
    \item 
    Every two different variables take identical values in $\mathscr{S}((f_{k+1})^{x_i=1})$ if and only if they take opposite values in $\mathscr{S}((f_{k+1})^{x_i=0})$, $i\in I(h_{k+1})$.
\end{enumerate}
     
     Because every two different columns of $f_k$ (as well as $f'_k$) do not take identical or opposite values by Corollary~\ref{cor: non-singularity}, we can easily check this proposition still holds for $f_{k+1}$. Next we show the signature $\widehat{f_{k+1}}$ is affine. We first notice $\widehat{f_{k+1}}$ has the following properties, which holds for every $i\in I(\widehat{f_{k+1}})$:
\begin{enumerate}
    \item $(\widehat{f_{k+1}})^{x_i=0}=(f_{k+1})^{x_i=0}=B_k\in \mathscr{A}$.
    \item 
    $(\widehat{f_{k+1}})^{x_i=1}=(\widehat{f_k})_{\times 2}\in \mathscr{A}$, since $f_k$ is the right wing of $B_k$ by induction hypothesis, hence $\widehat{f_k}$ is the right half of $B_k$, which is affine.
\end{enumerate}

We prove $\widehat{f_{k+1}}\in \mathscr A$ by contradiction. If $\widehat{f_{k+1}}\notin \mathscr{A}$, there exists $\alpha',\beta',\gamma'\in \mathscr{S}(\widehat{f_{k+1}})$, such that $\theta'=\alpha'\oplus\beta'\oplus\gamma'\notin\mathscr{S}(\widehat{f_{k+1}})$. If there exists a position $j$ such that $\alpha'_j=\beta'_j=\gamma'_j=0$, assume $j=1$ without loss of generality. Then $\alpha'=0\alpha,\beta'=0\beta,\gamma'=0\gamma,\theta'=0\theta$, where $\alpha,\beta,\gamma\in\mathscr{S}(\widehat{(f_{k+1}})_j^0)$ and $\theta=\alpha\oplus\beta\oplus\gamma$. Because $(\widehat{f_{k+1}})_j^0\in\mathscr{A}$, we have $\theta\in\mathscr{S}((\widehat{f_{k+1}})_j^0)$. Then $\theta'=0\theta\in\mathscr{S}(\widehat{f_{k+1}})$, contradiction! Similarly, if there exists a position $j$ such that $\alpha'_j=\beta'_j=\gamma'_j=1$, we also get a contradiction.

    We enumerate all possible combinations of $\alpha',\beta',\gamma'\in\mathscr{S}({\widehat{f_{k+1}}})$, but exclude cases in which they all come from $\mathscr{S}(({\widehat{f_{k+1}}})^{x_1=0})$ or $\mathscr{S}(({\widehat{f_{k+1}}})^{x_1=1})$, since both of them are affine. In the following proof, vectors $\vec{0}$ and $\vec{1}$ are short for $\vec{0}^{2^k}$ and $\vec{1}^{2^k}$.

    \begin{enumerate}
        \item $\alpha'=\vec{1}\vec{1},\beta'=\overline{\beta}\beta,\gamma'=\overline{\gamma}\gamma$, where $\beta,\gamma\in \mathscr{S}(f_k)$. Then the $(2^k+1)$-th position of $\alpha',\beta',\gamma'$ are all $1$s, contradiction.
        \item $\alpha'=\vec{1}\vec{1},\beta'=\overline{\beta}\beta,\gamma'=\gamma\gamma$, where $\beta\in 
        \mathscr{S}(f_k),\gamma\in \mathscr{S}(f'_k)$. Because $\alpha',\beta',\gamma'$ cannot take identical values in the same position, if $\beta_j=1$, $\gamma_j$ must be $0$; if $\beta_j=0,\overline{\beta_j}=1$, $\gamma_j$ must be $0$. Thus, $\gamma=\vec{0}$. This contradicts $f_k=R_k$ which is EO.
        \item 
        $\alpha'=\overline{\alpha}\alpha,\beta'=\overline{\beta}\beta,\gamma'=\gamma\gamma$, where $\alpha,\beta\in\mathscr{S}(f_k),\gamma\in \mathscr{S}(f'_k)$. If there exists $j$ such that $\alpha_j=\beta_j=x\in\mathbb{Z}_2$, then $\gamma_j$ must be $\overline{x}$. 
        However, in this case, we have $\overline{\alpha}_j=\overline{\beta}_j=\gamma_j=\overline{x}$, then $\alpha',\beta',\gamma'$ take identical values in the same position. 
        So $\alpha=\overline{\beta}$. 
        This contradicts Corollary~\ref{cor: non-singularity}.
        \item 
        $\alpha'=\overline{\alpha}\alpha,\beta'=\beta\beta,\gamma'=\gamma\gamma$, where $\alpha\in \mathscr{S}(f_k),\beta,\gamma\in \mathscr{S}(f'_k)$. If there exists $j$ such that $\beta_j=\gamma_j$, whatever $\alpha_j$ takes, there will be a position where $\alpha',\beta',\gamma'$ take identical values. If $\beta=\overline{\gamma}$, there will be two rows taking opposite values in $f'_k$. Again this contradicts Corollary~\ref{cor: non-singularity}.
    \end{enumerate}

    By the above discussion, we show $\widehat{f_{k+1}}\in \mathscr{A}$. Thus, the signature
\NiceMatrixOptions{xdots={line-style = standard}}
$$B'_{k+1}=\begin{NiceArray}{*{8}c}[hvlines]
\Block{1-4}{} 0& 0& \cdots& 0& \Block{1-4}{} 1& 1& \cdots& 1\\
\Block{5-4}{\overline{f_{k+1}}} & & & & \Block{5-4}{f_{k+1}} & & &\\
& & & & & & &\\
& & & & & & &\\
& & & & & & &\\
& & & & & & &\\
\end{NiceArray}$$
is affine. Notice $B'_{k+1}$ is of arity $2^{k+2}$ and support size $2^{k+1}$. 
Moreover, every two different columns of $B'_{k+1}$ do not take identical values, since every two different columns of $f_{k+1}$ do not take identical or opposite values. Thus, $B'_{k+1}$ is butterfly of $k+1$ free variables. Therefore, $f_{k+1}=R_{k+1}$, which is the right wing of $B_{k+1}$.
\end{proof}

So far by Lemma~\ref{butterfly_is_kernel} and Lemma~\ref{uniqueness} we prove the equivalence between basic kernels and wings of butterfly. Combining with Lemma~\ref{Hadamard_is_butterfly} we get the proof of Theorem~\ref{main-theorem}.

 \begin{example}
We give basic $\delta_1$-affine kernels of order 2, 3, and 4 constructed using Sylvester's method, where the all-1 vector in the first row is removed from the 1-Hadamard code.

\NiceMatrixOptions{xdots={line-style = dashed}}
$$f_{2}=
\begin{bNiceArray}{*{4}c}[margin,first-row]

\textcolor{gray}{1}& \textcolor{gray}{1}& \textcolor{gray}{1}& \textcolor{gray}{1}\\
1& 0& 1& 0\\
1& 1& 0& 0\\
1& 0& 0& 1\\

\end{bNiceArray}$$

\NiceMatrixOptions{xdots={line-style = standard}}
$$f_3=
\begin{bNiceArray}{*{4}c|*{4}c}[margin,first-row]
\textcolor{gray}{1}& \textcolor{gray}{1}& \textcolor{gray}{1}& \textcolor{gray}{1}& \textcolor{gray}{1}& \textcolor{gray}{1}& \textcolor{gray}{1}& \textcolor{gray}{1}\\
1& 0& 1& 0 & 1& 0& 1& 0\\
1& 1& 0& 0 & 1& 1& 0& 0\\
1& 0& 0& 1 & 1& 0& 0& 1\\\hline
1& 1& 1& 1 & 0& 0& 0& 0\\
1& 0& 1& 0 & 0& 1& 0& 1\\
1& 1& 0& 0 & 0& 0& 1& 1\\
1& 0& 0& 1 & 0& 1& 1& 0\\

\end{bNiceArray}$$

\NiceMatrixOptions{xdots={line-style = standard}}
$$f_4=
\begin{bNiceArray}{*{8}c|*{8}c}[margin,first-row]
\textcolor{gray}{1}& \textcolor{gray}{1}& \textcolor{gray}{1}& \textcolor{gray}{1}& \textcolor{gray}{1}& \textcolor{gray}{1}& \textcolor{gray}{1}& \textcolor{gray}{1}& \textcolor{gray}{1}& \textcolor{gray}{1}& \textcolor{gray}{1}& \textcolor{gray}{1}& \textcolor{gray}{1}& \textcolor{gray}{1}& \textcolor{gray}{1}& \textcolor{gray}{1}\\
1& 0& 1& 0 & 1& 0& 1& 0&   1& 0& 1& 0 & 1& 0& 1& 0\\
1& 1& 0& 0 & 1& 1& 0& 0&   1& 1& 0& 0 & 1& 1& 0& 0\\
1& 0& 0& 1 & 1& 0& 0& 1&   1& 0& 0& 1 & 1& 0& 0& 1\\
1& 1& 1& 1 & 0& 0& 0& 0&   1& 1& 1& 1 & 0& 0& 0& 0\\
1& 0& 1& 0 & 0& 1& 0& 1&   1& 0& 1& 0 & 0& 1& 0& 1\\
1& 1& 0& 0 & 0& 0& 1& 1&   1& 1& 0& 0 & 0& 0& 1& 1\\
1& 0& 0& 1 & 0& 1& 1& 0&   1& 0& 0& 1 & 0& 1& 1& 0\\\hline

1& 1& 1& 1 & 1& 1& 1& 1&   0& 0& 0& 0 & 0& 0& 0& 0\\
1& 0& 1& 0 & 1& 0& 1& 0&   0& 1& 0& 1 & 0& 1& 0& 1\\
1& 1& 0& 0 & 1& 1& 0& 0&   0& 0& 1& 1 & 0& 0& 1& 1\\
1& 0& 0& 1 & 1& 0& 0& 1&   0& 1& 1& 0 & 0& 1& 1& 0\\
1& 1& 1& 1 & 0& 0& 0& 0&   0& 0& 0& 0 & 1& 1& 1& 1\\
1& 0& 1& 0 & 0& 1& 0& 1&   0& 1& 0& 1 & 1& 0& 1& 0\\
1& 1& 0& 0 & 0& 0& 1& 1&   0& 0& 1& 1 & 1& 1& 0& 0\\
1& 0& 0& 1 & 0& 1& 1& 0&   0& 1& 1& 0 & 1& 0& 1& 1\\

\end{bNiceArray}$$
  
\end{example}

\section{Hardness result}\label{sec:hardness}
In this section, we show that \#EO$(f,g)$ is \#P-hard if $f\in \mathscr{D}_1\setminus\mathscr{A}$ and $g\in \mathscr{D}_0\setminus\mathscr{A}$ (Theorem~\ref{thm: EO f g is hard}), which directly implies Theorem~\ref{thm:hard}.
The proof relies on the characterization of $\delta_1$-affine and $\delta_0$-affine kernels and consists of three parts. 
First, from Lemma~\ref{interpolation} to Corollary~\ref{reduction-to-EO^c}, we show how to realize the binary signature $\delta_1\otimes\delta_0$ from an EO signature that does not satisfy the arrow reversal symmetry (ARS) property. 
Second, from Lemma~\ref{hardness_1} to Lemma~\ref{hardness_3}, we prove that \#EO$(f',g')$ is \#P-hard where $f'\in \mathscr{D}_1\setminus\mathscr{A}$ and $g'\in \mathscr{D}_0\setminus\mathscr{A}$ both have support size $3$. Finally, in Lemma~\ref{realize-extract-1} and Theorem~\ref{thm: EO f g is hard}, we show that every such pair of signatures $f'$ and $g'$ can be realized from arbitrary signatures $f\in \mathscr{D}_1\setminus\mathscr{A}$ and $g\in \mathscr{D}_0\setminus\mathscr{A}$, respectively, using the binary signature $\delta_1\otimes\delta_0$.

   Only in this section we require signatures taking real values. An EO signature $f$ of 2 arity is written as
    $f=\left[\begin{smallmatrix}
    0 & a\\
    b & 0\\
    \end{smallmatrix}\right]$, where $a,b\in \mathbb{R}$ if $f(01)=a,f(10)=b$.
    {\sc Disequality} is the 2 arity signature $(\neq_2)=\left[\begin{smallmatrix}
    0 & 1\\
    1 & 0\\
    \end{smallmatrix}\right]$;
    {\sc Minus Disequality} is the 2 arity signature
    $(\neq_2^-)=\left[\begin{smallmatrix}
    0 & 1\\
    -1 & 0\\
    \end{smallmatrix}\right]$.
    The symbol "$\leq_T$" denotes Turing reduction.
    
     \begin{definition}[ARS, -ARS]
    A signature $f$ taking real values satisfies ARS (or -ARS) if $f(\alpha)=f(\overline{\alpha})$ (or $f(\alpha)=-f(\overline{\alpha})$) for every $\alpha\in \mathscr S(f)$.
    \end{definition}
    
    \begin{lemma}\label{interpolation}
        Let $f=\left[\begin{smallmatrix}
    0 & 1\\
    x & 0\\
    \end{smallmatrix}\right]$ 
       be a binary signature taking real values, where $x\neq \pm 1$. Then \#EO$(f,\delta_1\otimes\delta_0)\leq_T$ \#EO$(f)$.
    \end{lemma}

    \begin{proof}
        This follows from a standard interpolation technique. See Lemma 2.3 in \cite{CFS21}.
    \end{proof}
     
    \begin{lemma}\label{ARS-lemma}
        Let $f$ be an \rm{EO} signature of arity $2n\geq 4$ taking real values. Then $f$ satisfies ARS (or -ARS) if and only if $f^{x_i\neq x_j}$ satisfies ARS (or -ARS), for every $i\neq j, i,j\in I(f)$.        
    \end{lemma}

    \begin{proof}
        We prove only the ARS version, as the -ARS version follows similarly.
        
        First, assume $f$ satisfies ARS. For any $\alpha\in\mathcal{H}_{2n-2}$ and $i,j$, we have $f^{x_i\neq x_j}(\alpha)=f_{ij}^{01}(\alpha)+f_{ij}^{10}(\alpha)$ and $f^{x_i\neq x_j}(\overline{\alpha})=f_{ij}^{01}(\overline{\alpha})+f_{ij}^{10}(\overline{\alpha})$. Since $f$ satisfies ARS, we have $f_{ij}^{01}(\alpha)=f_{ij}^{10}(\overline{\alpha})$ and $f_{ij}^{10}(\alpha)=f_{ij}^{01}(\overline{\alpha})$. Therefore, $f^{x_i\neq x_j}(\alpha)=f^{x_i\neq x_j}(\overline{\alpha})$, showing $f^{x_i\neq x_j}$ satisfies ARS.

        Conversely, assume $f^{x_i\neq x_j}$ satisfies ARS. We prove by contradiction that $f$ must satisfy ARS. Suppose there exists $\alpha\in \mathcal{H}_{2n}$ such that $f(\alpha)\neq f(\overline{\alpha})$. By relabeling variables, we may write $\alpha=001\theta,\overline{\alpha}=110\overline{\theta}$. Let $\beta=010\theta,\overline{\beta}=101\overline{\theta}$ and $\gamma=100\theta,\overline{\gamma}=011\overline{\theta}$. All these six vectors are in $\mathcal{H}_{2n}$. Then $f^{x_2\neq x_3}(0\theta)=f(\alpha)+f(\beta)$ and $f^{x_2\neq x_3}(1\overline{\theta})=f(\overline{\alpha})+f(\overline{\beta})$. Since $f^{x_2\neq x_3}$ satisfies ARS, we have $f(\alpha)+f(\beta)=f(\overline{\alpha})+f(\overline{\beta})$. Similarly, $f(\beta)+f(\gamma)=f(\overline{\beta})+f(\overline{\gamma})$ and $f(\gamma)+f(\alpha)=f(\overline{\gamma})+f(\overline{\alpha})$. These three equations imply $f(\alpha)=f(\overline{\alpha})$, contradicting our assumption.
    \end{proof}

     \begin{lemma}\label{ARS-ARS-lemma}
        Let $f$ be an \rm{EO} signature of arity $2n\geq 4$ taking real values. If for any $i\neq j, i,j\in I(f)$, both $f^{x_i\neq x_j}$ and $f^{x_i\neq^- x_j}$ satisfy either ARS or -ARS, then $f$ satisfies either ARS or -ARS.
    \end{lemma}

    \begin{proof}
        For any $i\neq j$ and $\beta\in \mathcal{H}_{2n-2}$, since $f^{x_i\neq x_j}$ satisfies ARS or -ARS, we have
        \begin{equation}
            (f_{ij}^{01}(\beta)+f_{ij}^{10}(\beta))^2=(f_{ij}^{01}(\overline{\beta})+f_{ij}^{10}(\overline{\beta}))^2
        \end{equation}
        Since $f^{x_i\neq^- x_j}$ satisfies ARS or -ARS, we have
        \begin{equation}
            (f_{ij}^{01}(\beta)-f_{ij}^{10}(\beta))^2=(f_{ij}^{01}(\overline{\beta})-f_{ij}^{10}(\overline{\beta}))^2
        \end{equation}
        Adding (1) and (2), we obtain
        \begin{equation}
            f_{ij}^{01}(\beta)^2+f_{ij}^{10}(\beta)^2=f_{ij}^{01}(\overline{\beta})^2+f_{ij}^{10}(\overline{\beta})^2
        \end{equation}
        Pinning 0 to any variable $x_k$ of $\beta=x_k\beta'$,
        \begin{equation}
            f_{ijk}^{010}(\beta')^2+f_{ijk}^{100}(\beta')^2=
            f_{ijk}^{101}(\overline{\beta'})^2+f_{ijk}^{011}(\overline{\beta'})^2
        \end{equation}
    By rotating $i,j,k$, we have
    \begin{equation}
        \begin{aligned}
            &f_{ijk}^{001}(\beta')^2+f_{ijk}^{010}(\beta')^2=
            f_{ijk}^{110}(\overline{\beta'})^2+f_{ijk}^{101}(\overline{\beta'})^2\\
            &f_{ijk}^{100}(\beta')^2+f_{ijk}^{001}(\beta')^2=
            f_{ijk}^{011}(\overline{\beta'})^2+f_{ijk}^{110}(\overline{\beta'})^2
        \end{aligned}
    \end{equation}
    From (4) and (5), we obtain
    \begin{equation}
    \begin{aligned}
        &f_{ijk}^{100}(\beta')^2=f_{ijk}^{011}(\overline{\beta'})^2\\
        &f_{ijk}^{010}(\beta')^2=f_{ijk}^{101}(\overline{\beta'})^2\\
        &f_{ijk}^{001}(\beta')^2=f_{ijk}^{110}(\overline{\beta'})^2\\
    \end{aligned}
    \end{equation}
    Since $i,j,k$ and $\beta$ are arbitrarily chosen, we have $f(\alpha)^2=f(\overline{\alpha})^2$, thus $f(\alpha)=\pm f(\overline{\alpha})$, for any $\alpha\in \mathcal{H}_{2n}$. 

    Next we show that the choice of $+$ or $-$ is uniform across all $\alpha$, i.e., either $f(\alpha)=f(\overline{\alpha})$ for every $\alpha\in \mathcal{H}_{2n}$ or $f(\alpha)=-f(\overline{\alpha})$ for every $\alpha\in \mathcal{H}_{2n}$. We prove by contradiction. If not, there exist $\alpha,\beta\in \mathcal{H}_{2n}$ such that $f(\alpha)\neq f(\overline{\alpha})$ and $f(\beta)\neq -f(\overline{\beta})$. Then $f(\alpha)=-f(\overline{\alpha})\neq 0$ and $f(\beta)=f(\overline{\beta})\neq0$. Choose such $\alpha,\beta$ with minimum Hamming distance. We claim wt$(\alpha\oplus\beta)=2$. Otherwise, suppose wt$(\alpha\oplus\beta)=2k,k>1$. We can write $\alpha=x_1x_2...x_{2k}\gamma$ and $\beta=\overline{x_1x_2...x_{2k}}\gamma$ by relabeling variables. Since $\alpha,\beta\in \mathcal{H}_{2n}$, we must have wt$(x_1x_2...x_{2k})=k$. We may assume $x_1\neq x_2$. Let $\alpha'=\overline{x_1}\overline{x_2}x_3...x_{2k}\gamma$ and $\beta'=x_1x_2\overline{x_3...x_{2k}}\gamma$, then $\alpha',\beta'\in \mathcal{H}_{2n}$. We must have $f(\alpha')=0$. Otherwise, either $f(\alpha')=f(\overline{\alpha'})\neq 0$ or $f(\alpha')=-f(\overline{\alpha'})\neq 0$. In the first case, $f(\alpha')\neq -f(\overline{\alpha'})$ and wt$(\alpha\oplus\alpha')=2<2k$, contradicting minimality. In the second case, $f(\alpha')\neq f(\overline{\alpha'})$ and wt$(\alpha'\oplus\beta)=2k-2<2k$, also contradicting minimality. Thus, $f(\alpha')=f(\overline{\alpha'})=0$ and similarly $f(\beta')=f(\overline{\beta'})=0$. Consider:
    \begin{equation}
        \begin{aligned}
            &f^{x_1\neq x_2}(x_3...x_{2k}\gamma)=f(\alpha)+f(\alpha')=f(\alpha)\\
            &f^{x_1\neq x_2}(\overline{x_3...x_{2k}\gamma})=f(\overline{\alpha})+f(\overline{\alpha'})=f(\overline{\alpha})\\
            &f^{x_1\neq x_2}(\overline{x_3...x_{2k}}\gamma)=f(\beta)+f(\beta')=f(\beta)\\
            &f^{x_1\neq x_2}(x_3...x_{2k}\overline{\gamma})=f(\overline{\beta})+f(\overline{\beta'})=f(\overline{\beta})\\
        \end{aligned}
    \end{equation}
    By assumption $f^{x_1\neq x_2}$ satisfies ARS or -ARS. If ARS, we have $f(\alpha)=f(\overline{\alpha})$ by (7), contradicting $f(\alpha)\neq f(\overline{\alpha})$.
    If -ARS, we have $f(\beta)=-f(\overline{\beta})$ by (7), contradicting $f(\beta)\neq -f(\overline{\beta})$. Thus wt$(\alpha\oplus\beta)=2$. Write $\alpha=01\theta$ and $\beta=10\theta$. Then $f^{x_1\neq x_2}(\theta)=f(\alpha)+f(\beta)$ and $f^{x_1\neq x_2}(\overline{\theta})=f(\overline{\alpha})+f(\overline{\beta})=f(\alpha)-f(\beta)$. Since $f^{x_1\neq x_2}$ satisfies ARS or -ARS, we get $(f(\alpha)+f(\beta))^2=(f(\alpha)-f(\beta))^2$, thus $f(\alpha)f(\beta)=0$, contradicting $f(\alpha)\neq0$ and $f(\beta)\neq 0$. Therefore $f(\alpha)=\pm f(\overline{\alpha})$ uniformly for every $\alpha\in \mathcal{H}_{2n}$, meaning $f$ satisfies either ARS or -ARS.
    \end{proof}
     
     \begin{lemma}
       Let $f$ be an \rm{EO} signature of arity $2n\geq2$ taking real values. If $f$ does not satisfy ARS or -ARS, then \#EO$(f,\delta_1\otimes\delta_0)\leq_T$ \#EO$(f)$.   
    \end{lemma}

    \begin{proof}
        If $n=1$, then \#EO$(f,\delta_1\otimes\delta_0)\leq_T$ \#EO$(f)$ by Lemma~\ref{interpolation}. If $n>1$, there exists $i\neq j$ such that $f^{x_i\neq x_j}$ does not satisfy ARS by Lemma~\ref{ARS-lemma}. Continuing this process, we realize a binary signature $f_2$ which does not satisfy ARS. Suppose $f_2=\left[\begin{smallmatrix}
        0 & 1\\
        x & 0\\
        \end{smallmatrix}\right]$ (up to a constant), where $x\neq 1$. 
        If $x\neq -1$, then it reduces to the case when $n=1$. If $x=-1$, we realize $\neq_2^{-}$. By Lemma~\ref{ARS-ARS-lemma}, there exists $i\neq j$ such that $f^{x_i\neq x_j}$ does not satisfy ARS or -ARS, or $f^{x_i\neq^- x_j}$ does not satisfy ARS or -ARS. Continuing to loop by $\neq_2$ or $\neq_2^-$, we realize a binary signature $f_2$ which does not satisfy ARS or -ARS. Then it reduces to the case when $n=1$.
    \end{proof}

    \begin{corollary}\label{reduction-to-EO^c}
        Let $f$ be an \rm{EO} signature of arity $2n\geq2$ taking 0-1 values. If $f$ does not satisfy ARS, then \#EO$(f,\delta_1\otimes\delta_0)\leq_T$ \#EO$(f)$. 
    \end{corollary}

    Next, we will prove \#EO$(f,g)$ is \#P-hard if $f\in \mathscr{D}_1\backslash\mathscr{A}$ and $g\in \mathscr{D}_0\backslash\mathscr{A}$.

    \begin{lemma}\label{hardness_1}
    
    Let 
        $$
    f_2=
    \begin{bNiceArray}{cccc}[first-row,margin]
         x_1 & x_2 & x_3 & x_4\\
         1 &1 & 0& 0\\
         1 &0 & 1& 0\\
         1 &0 & 0& 1
    \end{bNiceArray}
   , \quad 
    g_2=
    \begin{bNiceArray}{cccc}[first-row,margin]
         y_1 & y_2 & y_3 & y_4\\
         0 &0 & 1& 1\\
         0 &1 & 0& 1\\
         0 &1 & 1& 0
    \end{bNiceArray}
    $$
    then \#EO($f_2,g_2$) is \#P-hard.
    \end{lemma}
    
    \begin{proof}
    Notice \#EO($f_2,g_2$)$\equiv_T$Holant($\neq_2|f_2,g_2$). Loop $x_1$ and $y_1$, $x_2$ and $y_2$ using $\neq_2$, we realize
    $$
    h=
    \begin{bNiceArray}{cccc}[first-row,margin]
         x_3 & x_4 & y_3 & y_4\\
         0 &0 & 1& 1\\
         1 &0 & 0& 1\\
         1 &0 & 1& 0\\
         0 &1 & 0& 1\\
         0 &1 & 1& 0
    \end{bNiceArray}
    $$
    \#EO($h$) is \#P-hard by the six-vertex model \cite{cfx}.
     \end{proof}

     \begin{lemma}\label{hardness_2}
         Let $f=(f_2)_{\times a},g=(g_2)_{\times b}$, where $a,b\in \Z_+$. Then \#EO$(f,g)$ is \#P-hard.
     \end{lemma}

    \begin{proof}
        We loop $f$ and $g$ using $\neq_2$ block by block. That is, loop $x_i$ and $y_i$ for $1\leq i\leq 4$, where $x_i,y_i$ are variables of one block $f_2$ and $g_2$ as denoted in Lemma~\ref{hardness_1}. 

        If $a=b=1$, it is the case in Lemma~\ref{hardness_1}. If $a=b>1$, we can realize 
    $$
    h=
    \begin{bNiceArray}{cccccccc}[first-row,margin]
         x_1 & x_2 & x_3 & x_4 &y_1 & y_2 & y_3 & y_4\\
         1 &1 & 0& 0& 0 &0 & 1& 1\\
         1 &0 & 1& 0& 0 &1 & 0& 1\\
         1 &0 & 0& 1& 0 &1 & 1& 0
    \end{bNiceArray}
    $$
    Loop $x_1$ and $y_1$, $x_2$ and $y_2$, we realize 
    $$
    h'=
    \begin{bNiceArray}{cccccccc}[first-row,margin]
         x_3 & x_4 &y_3 & y_4\\
         0& 0& 1& 1\\
         1& 0& 0& 1\\
         0& 1& 1& 0
    \end{bNiceArray}
    $$
    which is \#P-hard by the six vertex model.

    If $a\neq b$, we may assume $a>b$ without loss of generality. Then we can realize $f'=(f_2)_{\times(a-b)}$. Then loop $f'$ and $g$ block by block and continue this process. We finally realize $(f_2)_{\times d}$ and $(g_2)_{\times d}$, where $d=gcd(a,b)$ is computed by the Euclidean algorithm. Then it reduces to the case $a=b$.
    \end{proof}

    \begin{lemma}\label{hardness_3}
         Let $f\in \mathscr{D}_1\backslash\mathscr{A},g\in\mathscr{D}_0\backslash\mathscr{A}$, and $|\mathscr{S}(f)|=|\mathscr{S}(g)|=3$. Then \#EO$(f,g)$ is \#P-hard.
    \end{lemma}

    \begin{proof}
    We may assume $f=\delta_1^{\otimes m}\otimes h$, where $m\in\Z_+$ and $\delta_1,\delta_0\nmid h$. Write $f$ as:
    $$
    f=
    \begin{bNiceArray}{ccccccc}[first-row,last-row,margin]
         x_1 & x_2 & x_3 & x_4 &y_2 & y_3 & y_4\\
         1 &1 & 0& 0 &0 & 1& 1\\
         1 &0 & 1& 0 &1 & 0& 1\\
         1 &0 & 0& 1 &1 & 1& 0\\
         m &a & b &c &d &e & f
    \end{bNiceArray}
    $$
    where $a,b,c,d,e,f$ denote the number of their corresponding columns in $\mathscr{S}(f)$. Because $f$ is an EO signature, we have:
    
    \begin{equation}
    \begin{aligned}
        &m+c+d+e=a+b+f\\
        &m+b+d+f=a+c+e\\
        &m+a+e+f=b+c+d\\
    \end{aligned}
    \end{equation}
    This gives us
    \begin{equation}
    \begin{aligned}
        &a=m+d\\
        &b=m+e\\
        &c=m+f\\
    \end{aligned}
    \end{equation}
    Then we loop $x_i$ and $y_i$ for $2\leq i\leq 4$ and realize $(f_2)_{\times m}$. Similarly, we can realize multiple of $g_2$ from $g$. Then by Lemma~\ref{hardness_2} we prove \#EO$(f,g)$ is \#P-hard.
    \end{proof}

    \begin{lemma}\label{realize-extract-1}
        Let $f$ be an \rm{EO} signature and $f$ does not satisfy ARS. Then $$\#{\rm EO}(\{f,f^{x_i=1},f^{x_i=0}\})\leq_T\#{\rm EO}(f),$$ for every $i\in I(f)$.
    \end{lemma}
    \begin{proof}
         Because $f$ does not satisfy ARS, we first realize $\delta_0\otimes\delta_1$ by Corollary~\ref{reduction-to-EO^c}. Then loop $\delta_0$ and $x_i$ with $\neq_2$, we realize $f^{x_i=1}$.
         Similarly we can realize $f^{x_i=0}$.
    \end{proof}

   \begin{theorem}\label{thm: EO f g is hard}
        If $f\in\mathscr{D}_1\backslash \mathscr{A},g\in\mathscr{D}_0\backslash\mathscr{A}$, then \#EO($f,g$) is \#P-hard.
   \end{theorem}
   
   \begin{proof}
       We first show we can realize a signature $f'\in \mathscr{D}_1\backslash\mathscr{A}$ of support size 3 if $|\mathscr{S}(f)|>3$.
       If $f\in K(\mathscr{D}_1)$, suppose $f=(f_k)_{\times m}$ by Lemma~\ref{multiple_kernel}, where $f_k$ is the $\delta_1$-affine kernel of order $k$. We can extract 1 in some variable $x_i$ of $f$ by Lemma~\ref{realize-extract-1} and realize $f^{x_i=1}$, which is a multiple of $f_{k-1}$ by Lemma~\ref{long_lemma}. Continuing this process we can finally realize multiple of $f_2$.
       If $f\notin K(\mathscr{D}_1)$, suppose $f=\delta_1\otimes h$. Then there exists $i\in I(h)$ such that $h_i^0\in\mathscr{D}_1\backslash\mathscr{A}$. We have $3<|\mathscr{S}(h_i^0)|<|\mathscr{S}(f)|$. Repeating the process, either $h_i^0\in K(\mathscr{D}_1)$ or it can realize a signature in $\mathscr{D}_1\backslash\mathscr{A}$ of smaller support size. Finally we can realize $f'\in\mathscr{D}_1\backslash\mathscr{A}$ of support size 3. Similarly, we can realize $g'\in\mathscr{D}_0\backslash\mathscr{A}$ of support size 3. Then by Lemma~\ref{hardness_3} we prove \#EO$(f,g)$ is \#P-hard.
   \end{proof}
   
    \begin{remark}
    Although $\mathscr{D}_1$ or $\mathscr{D}_0$ alone is a tractable class, this example shows their mixing may be hard. This is because when looping $\delta_1$ and $\delta_0$, no new $\delta_1$ or $\delta_0$ would occur, so the chain reaction terminates. This is analogous to electron-positron annihilation in physics.
    \end{remark}

\bibliographystyle{elsarticle-num} 
\bibliography{eo}


\end{document}